\documentclass[twoside,leqno,twocolumn]{article}

\usepackage[letterpaper]{geometry}

\usepackage{siamproceedings}

\usepackage[T1]{fontenc}
\usepackage{amsfonts}
\usepackage{graphicx}
\usepackage{epstopdf}
\usepackage{enumitem}
\usepackage{algorithmic}
\ifpdf
  \DeclareGraphicsExtensions{.eps,.pdf,.png,.jpg}
\else
  \DeclareGraphicsExtensions{.eps}
\fi

\usepackage{booktabs} 					% For tables
\usepackage{threeparttable} 			% For tables
\usepackage{siunitx} 					% For thousand separators and percent signs

\newsiamremark{remark}{Remark}
\newsiamremark{hypothesis}{Hypothesis}
\crefname{hypothesis}{Hypothesis}{Hypotheses}
\newsiamthm{claim}{Claim}

\usepackage{amssymb}					% provinding the \supsetneq command
\usepackage{thm-restate}				% provinding the 'restatable' environment
\usepackage{xfrac}						% provinding the \sfrac command
\usepackage{xspace}						% provinding the \xspace command

\newcommand*{\nwspace}{\hspace*{.1em}}

\newcommand*{\Fyp}{\mathcal{F}}
\newcommand*{\Gyp}{\mathcal{G}}
\newcommand*{\Hyp}{\mathcal{H}}
\newcommand*{\Syp}{\mathcal{S}}
\DeclareMathOperator{\unhit}{unhit}
\DeclareMathOperator{\Tr}{Tr}

\newcommand*{\NP}{\textsf{NP}}
\newcommand*{\W}{\textsf{W}}
\newcommand*{\poly}{\textup{\textsf{poly}}}

\usepackage{tikz}
\usetikzlibrary{calc,shapes.geometric}
\tikzset{
    seqnode/.style={
        circle, 
        draw=black, 
        line width=1pt, 
        minimum size=2\seqradius, 
        inner sep=0pt, 
        outer sep=0pt
    },
    drawthree/selected w/.store in=\dtSelectedW,
    drawthree/selected w={},
    drawthree/selected u/.store in=\dtSelectedU,
    drawthree/selected u={},
}
\newlength{\seqradius}
\usepackage[dvipsnames]{xcolor}
\definecolor{conceptblue}{HTML}{6A91FF}
\colorlet{conceptlightblue}{conceptblue!30!white}

\begin{document}

\newcommand\relatedversion{}

\title{\Large Minimal-to-Maximal Conversion Search Is Not Output-Polynomial\relatedversion}
 \author{Bennet Hörmann\thanks{Karlsruhe Institute of Technology, Karlsruhe, Germany \email{first.lastname@kit.edu}.}   
    \and Martin Schirneck\footnotemark[1]}

\date{}

\maketitle

\fancyfoot[R]{\scriptsize{Copyright \textcopyright\ 2026\\
Copyright for this paper is retained by authors}}

%\pagenumbering{arabic}
%\setcounter{page}{1}%Leave this line commented out.

\begin{abstract} 
	The \emph{Transversal Hypergraph} problem is to enumerate
	(list)
	\emph{all} inclusion-wise minimal hitting sets of a given hypergraph $\Hyp$. 
	It is the most important open question in enumeration whether this problem
	admits an \emph{output-polynomial} algorithm
	whose running time scales polynomially with the size of $\Hyp$ 
	and the number of solutions.
	Currently, \emph{Minimal-to-Maximal Conversion Search} (MMCS) by Murakami and Uno [DAM 2014]
	is the most efficient algorithm for real-world instances,
	but there are no worst-case performance guarantees known for it.
	We prove that MMCS is in fact \emph{not} output-polynomial.
	The lower bound construction motivates a more detailed analysis of 
	how certain heuristic choices in the algorithm design affect the running time.
	We conduct a thorough analysis of those heuristics
	and, based on this, propose new extension.
	We then show in extensive running time experiments that 
	this new heuristic further improves practical performance.
	
	The code and testing data is available at \href{https://anonymous.4open.science/r/mmcs-is-not-output-polynomial-1F10}{anonymous.4open.science/r/mmcs-is-not-output-polynomial-1F10}.
\end{abstract}

\section{Introduction.}
\label{sec:intro}

Data scientists protecting private information in transaction databases~\cite{Stavropoulos16FrequentItemsetHiding},
computational biologists identifying proteins for targeted
treatments~\cite{Marazzi20OCSANA+,VeraLicona13OCSANA},
city planners deploying roadside units for autonomous vehicle routing~\cite{YefernyAllani18VehicleAdHocNet}.
They all face the same task:
given a collection of overlapping groups of entities, 
choose a small set of those entities such that each group has at least one representative selected.
In the bioinformatics example above, the input is a protein-protein interaction network
with a distinguished source and target.
The entities are the nodes of the network and each group corresponds to a source-target path.
The desired solutions are
``combinations of interventions''~\cite{VeraLicona13OCSANA},
sets of proteins such that inhibiting them destroys all paths.
It is often not enough to find a single representative set.
Instead, many applications require a diverse choice of multiple options
and sometimes all possible results.
This allows one to involve human experts, optimize secondary criteria in postprocessing,
and get a complete overview of the whole solution space.
The combinations of interventions in the protein networks, for example, are subsequently screened for potential side effects.

The collection of overlapping groups can be modeled as a \emph{hypergraph} $(V,\Hyp)$,
consisting of a finite set $V$ of \emph{vertices} together with a family $\Hyp$ of subsets of $V$, called the \emph{(hyper-)edges}.\footnote{
	We assume that $\Hyp$ contains at least one edge
	or, equivalently, any minimal hitting set contains at least one vertex.
}
Hypergraphs are a generalization of graphs where edges 
may have more (or fewer) than exactly two vertices.
They represent higher-order relations between entities.
A \emph{hitting set}, or \emph{transversal}, of $\Hyp$ is a set of vertices
that has a non-empty intersection with every edge of $\Hyp$.
A hitting set is \emph{(inclusion-wise) minimal} if it does not properly contain another hitting set.
These are the target solutions that need to be computed.
The collection of all minimal hitting sets is a new hypergraph on the same vertex set $V$,
called the \emph{transversal hypergraph} $\Tr(\Hyp)$.

The \emph{Tranversal Hypergraph} problem\footnote{%
	The problem is also known as the dualization of monotone Boolean functions 
	or hypergraph dualization.
	It is equivalent to enumerating the maximal independent sets of a hypergraph~\cite{Berge89Hypergraphs}
	or the minimal dominating sets of a graph~\cite{Kante14EnumerationMinimalDominatingSets}.
}
is the task, given the hypergraph $\Hyp$,
to generate all edges of $\Tr(\Hyp)$.
It is an example of an enumeration problem
involving the listing of \emph{all} solutions of a combinatorial structure.
This needs to be contrasted with finding a minimum-cardinality hitting set,
which is classically \NP-hard~\cite{Karp72Reducibility},
or merely counting all solutions.
The size of $\Tr(\Hyp)$ can be exponential in that of $\Hyp$,
which rules out a polynomial-time algorithm in the classical sense.
It is more appropriate to measure the running time of an enumeration algorithm
in terms of the combined input and output size,
that is, the number of vertices $n = |V|$, edges $m = |\Hyp|$, 
and minimal hitting sets $M = |\Tr(\Hyp)|$.
Even then, the exact complexity of the Transversal Hypergraph problem is unknown.
The current best upper bound of $N^{O(\sfrac{\log N}{\log\log N})}$, where $N = mn + M$,
was given by Fredman and Khachiyan~\cite{FredmanKhachiyan96Dualization}.
Whether this can be improved to \emph{output-polynomial} time $\textsf{poly}(N) = \textsf{poly}(n,m,M)$ is \emph{the}
major open question in enumeration since its inception in 1987~\cite{DemetrovicsThi87Antikeys,Mannila87Dependency,Reiter87DiagnosisFirstPrinciples}.

Practical applications of hitting set enumeration
like the protein-protein interaction mentioned above and many more~\cite{Birnick20HPIValid,Bleifuss24FDsThroughHittingSets,%
Livshits20ApproximateDenialConstraints,Marazzi20OCSANA+,VeraLicona13OCSANA,%
Xiao22DynamicFDs,YefernyAllani18VehicleAdHocNet}
use the \emph{Minimal-to-Maximal Conversion Search} (MMCS) algorithm
by Murakami and Uno~\cite{Murakami14Dualization}.
Experimental studies have shown it to be the current-fastest method on practical data~\cite{Gainer17AlgorithmsAndComputation}.
However, no worst-case bounds are known for its running time.
In particular, it is unknown whether MMCS is output-polynomial.

The algorithm is based on the fact that a hitting set is minimal
if and only if it is irredundant.
A set $S$ is called \emph{irredundant} if each of its vertices $s \in S$
has a \emph{private}, or \emph{critical}, hyperedge
$E_s \in \Hyp$ that is only hit by $s$,
meaning $E_s \cap S = \{s\}$.
MMCS grows a search tree in which each node is associated
with a partial \emph{solution} $S$ and a set $C$ of \emph{candidates} to extend $S$.
While $S$ is irredundant but not yet a hitting set,
the algorithm selects an unhit edge $E$
and branches on the decision which candidate from $E \cap C$ to add.
Once $S$ becomes redundant or a hitting set,
the respective branch is pruned.
(See \Cref{sec:prelim_MMCS} for a detailed description.)

MMCS employs two heuristics to reduce the size of the search tree.
The first one aims to keep the number of child nodes as small as possible.
\begin{itemize}
	\item \textit{Min-heuristic}: Whenever MMCS selects a branching edge $E$,
		it does so by minimizing the cardinality $|E \cap C|$ 
		over all unhit edges.
\end{itemize}
For the second heuristic, we need the concept of a violator.
Let $S$ be a partial solution,
$E$ the branching edge, and $v \in E \cap C$ a branching candidate.
If $v$ causes the new solution\footnote{%
	We use the notation $A \sqcup B$ for the set union
	to emphasize that $A$ and $B$ are disjoint.
}
$S \sqcup \{v\}$ of the child node 
to become redundant or a hitting set, we say $v$ is a \emph{violator}.
Note that adding $v$ to any superset $S' \supsetneq S$
makes $S' \cup \{v\}$ redundant as well.
Using $v$ in this subtree can never lead to a minimal solution.
\begin{itemize}
	\item \textit{Violator pruning}: Violators are removed from the candidate sets
	of all child nodes.
\end{itemize}

In its default configuration, MMCS uses both techniques simultaneously.
Murakami and Uno state that
``\emph{the time for the [heuristic] choice is usually shorter than the time saved by the choice,
thus we can reduce the total computation time}'' \cite[p.~86]{Murakami14Dualization}.
As it turns out, these choices have a tremendous impact on the asymptotic running time
on different families of hypergraphs.
We construct two sets of instances, one for each heuristic,
so that MMCS in default configuration solves them easily in output-polynomial time.
We can even choose these hypergraphs such that they only
have polynomially many solutions, $M \,{=}\, \poly(n,m)$.
That means, the overall running time is polynomial.
However, if the respective heuristic is turned off,
MMCS suddenly requires exponential time.

\begin{restatable}{theorem}{oldheuristics}
\label{thm:old_heuristics}
	For the min-heuristic and violator pruning, respectively,
	there exists a family of hypergraphs with $n$ vertices, $m$ hyperedges,
	and polynomially many minimal hitting sets for which \textup{MMCS} in default configuration
	enumerates all minimal hitting sets in (output-)polynomial time,
	but without the heuristic it requires $2^{\Omega(n)} \,{\cdot}\, \poly(m)$ time.
\end{restatable}

This still leaves the possiblity that MMCS in default configuration is in fact output-polynomial.
Our next result establishes that it is not.
Maybe surprisingly, we can even choose the counterexamples as mere graphs 
(all edges have two vertices).
The hitting sets of graphs are better known as vertex covers.
There are output-polynomial algorithms for the enumeration of 
minimal vertex covers~\cite{Johnson88MaxIndSet,Tsukiyama77GeneratingAllTheMaxIndSets},
but despite its practical efficiency, MMCS is unable to match this bound.

\begin{restatable}{theorem}{notoutputpoly}
\label{thm:not_output-poly}
	There exists a family of hypergraphs with $n$ vertices, $m$ hyperedges,
	and polynomially many minimal hitting sets on which \textup{MMCS} in default configuration takes
	$2^{\Omega(n)} \,{\cdot}\, \poly(m)$ time.
	Moreover, \textup{MMCS} cannot enumerate the minimal vertex covers of graphs
	in output-polynomial time.
\end{restatable}

Here, we assume the worst-case processing order of the branching candidates 
$E \cap C = \{v_1, v_2, \dots, v_k\}$.
This order is not specified in~\cite{Murakami14Dualization},
which motivates an in-depth study of how the order influences performance.
We suggest a new heuristic, which we call \emph{unhit degree sorting}.
The \emph{unhit degree} of a vertex is the number of yet unhit edges that contain this vertex.
The branching candidates are now processed in order of increasing unhit degree.
This turns out to be equivalent to minimizing the number of \emph{potential branching points} 
over all child nodes of the current node.
A potential branching point is an edge-vertex pair $(F,v)$
so that, if the child nodes select the edge $F$ for branching, 
then a node is created in the next level that corresponds to $v$.
Therefore, unhit degree sorting helps to keep the number of grandchildren of a node small.

We prove that this heuristic resolves the instances from \Cref{thm:not_output-poly},
but that one can still adapt the counterexample to force exponential running time.
Whether there exists a heuristic extension of MMCS that is provably output-polynomial
remains an open question.

\begin{restatable}{theorem}{adaptedcounterexample}
\label{thm:adapted_counterexample}
	\textup{MMCS} in default configuration and using unhit degree sorting is not output-polynomial.
\end{restatable}

Beyond asymptotic running times,
we also investigate whether the new heuristic improves
the empirical performance on practical instances.
As expected, unhit degree sorting decreases the size of the search tree on average.
However, this is counter-acted by the time needed to evaluate the unhit degree,
which is dynamically changing over the execution.
The resulting savings in running time are noticable but not very extensive.

We also run experiments using the \emph{global degree heuristic},
that employs the regular degree of a vertex as a static proxy.
This improves the running time significantly.
It becomes much lower than what can be explained by the (moderately) smaller tree size.
We identify another reason for the good performance.
The heuristic decreases the average degree over all vertices that are added/removed
to/from the partial solutions.
This causes MMCS to spend less time in each node of the search tree
conducting the redundancy check or evaluating whether the partial solution is a hitting set.

Our theoretical and empirical analysis reveals the outsized influence of heuristic choices 
in the algorithm design of MMCS.
This marks a step forward in the search for output-polynomial algorithms
for the Transversal Hypergraph problem.
We also believe that our results can form the foundation for further practical improvements in many application areas.
\vspace*{.25em}

\textbf{Disclosure of AI use.} 
Generative artificial intelligence, mainly Gemini 3.1 Pro and ChatGPT 5.5 Pro, was used to generate
the code for the figures in \Cref{sec:old_heuristics,sec:new_heuristic} as well as the scripts to coordinate and time the experiments in \Cref{sec:eval}.
The latter model was also used for a second verification of the adapted counterexample in \Cref{thm:adapted_counterexample}.
The initial construction as well as all of the writing is by the authors.
\vspace*{.25em}

\textbf{Outline.}
The next section reviews the details of the MMCS algorithm and how the min-heuristic and violator pruning are implemented.
We investigate the influence of both techniques in \Cref{sec:old_heuristics}.
The new unhit degree heuristic is introduced and theoretically analyzed in \Cref{sec:new_heuristic}.
Finally, we demonstrate the practical improvements of this addition
with an experimental analysis and evaluation in \Cref{sec:eval}.

\section{Minimal-to-Maximal Conversion Search.}
\label{sec:prelim_MMCS}

MMCS working on a hypergraph $(V,\Hyp)$ constructs a search tree
in which each node is labeled with a pair $(S,C)$
of disjoint sets $S \sqcup C \subseteq V$, see \Cref{fig:MMMCS_tree}.
The set $S$ is a partial solution and $C$ contains the candidate vertices.
The label of the root is $(\emptyset,V)$.
When expanding a node $(S,C)$, the algorithm verifies whether $S$ is irredundant.
Otherwise, $(S,C)$ is pruned and the search backtracks.
If $S$ is indeed irredundant, it is checked whether $S$ intersects every edge of $\Hyp$.
If so, $S$ is output since an irredundant hitting set is minimal.

\begin{figure}
	\centering
	\includegraphics{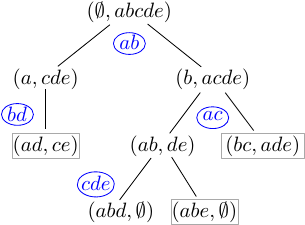}
	\caption{MMCS search tree for hypergraph $\Hyp = \{\{a,b\}, \{a,c\}, \{b,d\}, \{c,d,e\}\}$.
		The branching edges (blue) are chosen via the min-heuristic.
		Output nodes are framed.
		The partial solution $S = \{a,b,d\}$ is redundant, hence the violator $d$ 
		is not included as a candidate for the sibling with partial solution $\{a,b,e\}$.}
\label{fig:MMMCS_tree}
\end{figure}

If there exists an edge $E \in \Hyp$ that is not hit by $S$,
MMCS branches on the decision which candidate vertex from 
$E \cap C = \{v_1, \dots, v_k\}$ to include in $S$.
In the default configuration, MMCS selects $E$ using the min-heuristic, 
i.e., by minimizing $|E \cap C|$.
The algorithm then explores the child nodes $(S_i,C_i)$.
The new partial solution that corresponds to branching candidate $v_i$ is $S_i = S \sqcup \{v_i\}$.
In principle, the $i$-th candidate set is $C{\setminus}\{v_i,\dots,v_k\}$.
To also implement the violator pruning, 
the child nodes are processed in order of increasing index $i$.
Recall that we say a vertex $v_i$ is a violator if $S_i$ is redundant or a (minimal) hitting set.
Removing $v_i$ from the candidate of the later siblings results in an inductive definition: 
$C_1 = C{\setminus}E$,
and $C_{i+1} = C_i$ if $v_i$ is a violator, and $C_{i+1} = C_i \sqcup \{v_i\}$ otherwise.
Observe that, even though violator pruning is implemented inductively,
in effect a violator is removed from the candidate set of \emph{all} children 
(not only the later ones):
$v_i$ is not a candidate in any child node with index at most $i$
since it is contained in the branching edge $E$.

\section{Min-Heuristic and Violator Pruning.}
\label{sec:old_heuristics}

Consider MMCS working on an input hypergraph $\Hyp$ and let $(S,C)$ be a node in the search space.
If we could decide whether there exists a minimal hitting set $T \in \Tr(\Hyp)$
such that $S \subseteq T \subseteq S \cup C$,
this would allow for perfect pruning of the subtree rooted at $(S,C)$.
Unfortunately, the decision problem is \NP-complete~\cite{Boros98Subimplicants}
as well as $\W[3]$-complete when parameterized by the size $|S|$ of the partial solution~\cite{Blaesius22EfficientlyJCSS}.
MMCS avoids this by exploiting the weaker condition that 
if $S$ is contained in a minimal solution $T$, then $S$ must be irredundant.
Hence, the algorithm extends the set $S$ with more and more candidates until
it either becomes a hitting set (i.e., $S \in \Tr(\Hyp)$),
or until $S$ becomes redundant witnessing 
that there are no more solutions in this subtree.

The redundancy check is much faster than solving an \NP-hard problem in each step,
but the approach bears the danger that unnecessary branches are explored.
The heuristic design choices of MMCS aim 
to reduce the size of the search tree.
They influence which edge $E$ is selected for branching and which candidates $C$
are available to extend the current partial solution.
We investigate here how these choices affect the algorithm's behavior
on different families of hypergraphs.

\subsection{Selecting the Branching Edge.}
\label{subsec:old_min}

Consider a partial solution $S$ that is not yet a hitting set.
Hence, the set of unhit edges $\unhit(S) = \{E \in \Hyp \mid E \cap S  = \emptyset\}$
is non-empty.
Any minimal hitting set $T$ extending $S$ (if there is any) 
must contain a vertex from each such edge.
MMCS selects some $E \in \unhit(S)$ as the branching edge.
Not every vertex $E$ is considered to extend $S$.
Instead, each node in the search tree is also associated with a set $C$ of candidates.
This is done to avoid considering the same superset of $S$ in different subtrees
that only differ by the order in which the vertices are added.
Using candidate vertices is a common technique 
when enumerating the maximal sets of a given set system, 
see e.g.\ \cite{Tomita06GeneratingAllMaximalCliques}.
MMCS selects the branching edge by considering, for each $E \in \unhit(S)$,
how many child nodes would be created, i.e., $|E \cap C|$,
and chooses one that minimizes this property.

Besides $S$ being irredundant,
another necessary condition for the existence of a minimal hitting set $T$ 
with $S \subseteq T \subseteq S \cup C$,
is that $S \cup C$ is a hitting set for $\Hyp$ (which may itself not be minimal).
We first show that the min-heuristic, albeit not designed for that purpose,
also ensures that this condition is fulfilled.

\begin{restatable}{lemma}{minheuristiccandidates}
\label[lemma]{lem:min-heuristic_candidates}
	Consider \textup{MMCS} working on input $(V,\Hyp)$ using the min-heuristic.
	For every node $(S,C)$ in the search tree,
	$S \cup C$ is a hitting set for $\Hyp$.
\end{restatable}

\begin{proof}
	The claim of the lemma is equivalent to every edge $F \in \unhit(S)$ containing at least one candidate from $C$.
	We prove this by induction over the search tree.
 	It is clear for the candidate set $V$ in the root $(\emptyset,V)$.
  
  	Let $(S,C)$ be the current node and $E \in \unhit(S)$ the selected branching edge.
  	By induction, there exists a branching candidate $v \,{\in}\, E \cap C$.
	Let $(S \sqcup \{v\}, C_v)$ be the resulting child node.
	To reach a contradiction, we assume it has an edge $F \in \unhit(S \sqcup \{v\})$
	that is not hit by any of the remaining candidates, $F \cap C_v = \emptyset$.
	In other words, the set $C_0 = C{\setminus}C_v$ contains $F \cap C$ entirely.
	$C_0$ is always a subset of $E \cap C$ by the definition of MMCS
	(regardless of whether violator pruning is used).
	We can conclude $F \cap C \subseteq C_0 \subseteq E \cap C$.
	Since $F \in \unhit(S \sqcup \{v\})$, we know that $v \notin F$,
	and that $F$ is also in $\unhit(S)$.
	Therefore, $F \cap C$ is a \emph{proper} subset of  $E \cap C$
	and $|E \cap C|$ is not minimum, contradicting the min-heuristic.
\end{proof}

We show in the next subsection (\Cref{cor:min-heurisitc_inheritance})
that the min-heuristic also implies that the first branch explored by MMCS 
always results in a minimal solution.
If the min-heuristic is not used, it could happen that too many candidates are removed
in the child nodes, so that eventually $S \cup C$ is not a hitting set for $\Hyp$ anymore.
MMCS notices this fact (the latest) when $C$ runs empty while $\unhit(S)$ still contains some edges.
In this case, the subtree rooted at node $(S,C)$ is pruned.

We are now proving the part of \Cref{thm:old_heuristics} that pertains to the min-heuristic.
We construct a family of hypergraphs $\Hyp_{\min}$, one for each positive integer $\ell$,
that have $2\ell+1$ vertices and $\ell+1$ edges.
The hypergraphs all have only a \emph{single} minimal hitting set,
which MMCS finds instantly when using the min-heuristic 
(regardless of whether violator pruning is used).
However, if the algorithm uses any other branching edge than the one minimizing $|E \cap C|$,
it requires exponential time.

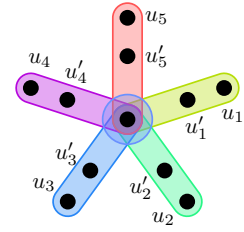
\begin{figure}
	\begin{center}
	\scalebox{0.85}{%
	\begin{tikzpicture}[baseline=-0.3em]
        \def\numvertices{5}
%        \pgfmathtruncatemacro{\numverticesMinusOne}{\numvertices - 1}
        
        % --- STYLE SWITCH ---
        % Toggle this to \grayedgesfalse to return to the rainbow colors
        \newif\ifgrayedges
        \grayedgesfalse
        % --------------------

        % 1. Main Radii for the rings
        \def\midRad{2.8}   % u'_i ring
        \def\outerRad{4.5} % u_i ring
        
        % 2. Shape Dimensions matched EXACTLY to your previous pill styles
        \def\edgeRadius{6.5}  % 6.5pt radius = 13pt total width
        \def\wEdgeRadius{11.0} % Slightly larger to encompass the w node neatly
        \def\borderWidth{0.75pt} 
        
        % 3. Central node coordinate
        \coordinate (w) at (0,0);
        
        % Pre-calculate baseline coordinates (1-based)
        \foreach \i in {1,...,\numvertices} {
            \pgfmathsetmacro{\ang}{90 - \i * 360 / \numvertices}
            \coordinate (up\i) at (\ang:\midRad em);
            \coordinate (u\i) at (\ang:\outerRad em);
        }

        % 4. Draw the l hyperedges {u_i, u'_i, w}
        \foreach \i in {1,...,\numvertices} {
            
            \ifgrayedges
                % Gray style based on the reference code
                \colorlet{edgeborder}{black}
                \colorlet{edgefill}{gray}
            \else
                % RGB Rainbow Color Calculation
                \pgfmathsetmacro{\hueAngle}{\i * 360 / \numvertices}
                \pgfmathsetmacro{\valR}{((cos(\hueAngle - 0) + 1) / 2)^0.8}
                \pgfmathsetmacro{\valG}{((cos(\hueAngle - 120) + 1) / 2)^0.8}
                \pgfmathsetmacro{\valB}{((cos(\hueAngle - 240) + 1) / 2)^0.8}
                \definecolor{edgecolor}{rgb}{\valR, \valG, \valB}
                
                \colorlet{edgeborder}{edgecolor}
                \colorlet{edgefill}{edgecolor}
            \fi
            
            \pgfmathsetmacro{\ang}{90 - \i * 360 / \numvertices}
            
            % Draw the perfectly calculated 2D pill shape!
            \filldraw[
                fill=edgefill, fill opacity=0.35,
                draw=edgeborder, draw opacity=1.0,
                line width=\borderWidth, line join=round
            ] 
                % Start right side of the origin
                ($ (w) + ({\ang-90}:\edgeRadius pt) $)
                
                % Straight line to the right side of the outer node
                -- ($ (u\i) + ({\ang-90}:\edgeRadius pt) $)
                
                % Semicircle cap wrapping over the top of the outer node
                arc[start angle={\ang-90}, delta angle=180, radius=\edgeRadius pt]
                
                % Straight line back down the left side to the origin
                -- ($ (w) + ({\ang+90}:\edgeRadius pt) $)
                
                % Semicircle cap wrapping around the bottom of the origin
                arc[start angle={\ang+90}, delta angle=180, radius=\edgeRadius pt]
                
                -- cycle;
        }

        % 5. Draw the single hyperedge containing only {w}
        % Drawn in the foreground with true transparency
        \filldraw[
            fill=conceptblue, fill opacity=0.35,
            draw=conceptblue, draw opacity=1.0,
            line width=\borderWidth
        ] (w) circle (\wEdgeRadius pt);

        % 6. Draw all nodes and labels ON TOP
        
        % w node
        \node[seqnode, fill=black] (node_w) at (w) {};
        % Positioned diagonally so it clears the overlapping opaque borders
        % \node[anchor=south east] at ($ (w) + (135:0.6em) $) {$w$};
        
        \foreach \i in {1,...,\numvertices} {
            \pgfmathsetmacro{\ang}{90 - \i * 360 / \numvertices}
            
            % u'_i nodes (middle ring)
            \node[seqnode, fill=black] (node_up\i) at (up\i) {};
            \node[anchor=center] at ($ (up\i) + (\ang-90:1.3em) $) {$u'_{\i}$};
            
            % u_i nodes (outer ring)
            \node[seqnode, fill=black] (node_u\i) at (u\i) {};
            \node[anchor=center] at ($ (u\i) + (\ang-90:1.3em) $) {$u_{\i}$};
        }
    \end{tikzpicture}
    } % end of \scalebox
    \vspace*{-1.25em}
    \end{center}
\caption{The hypergraph $\Hyp_{\min}$ for parameter $\ell = 5$.
	The vertex $w$ is in the center.}
\label{fig:min-heuristic}
\end{figure}

The vertex set is $V_{\min} = \{w, u_1, \dots, u_{\ell}, u'_1, \dots u'_\ell\}$ and
the edges are
\begin{equation*}
	\Hyp_{\min} = \{ \{w, u_j, u'_j\}  \mid 1 \le j \le \ell \} \sqcup \{ \{w\} \}.
\end{equation*}
See \Cref{fig:min-heuristic} for a visualization.
Note that $\{w\}$ is the only inclusion-wise minimal edge.
It is well-known that for the transversal hypergraph only the minimal edges are relevant.
In more detail, let $\min(\Hyp)$ be the subhypergraph of minimal edges,
we then have $\Tr(\Hyp) = \Tr(\min(\Hyp))$, see~\cite{Berge89Hypergraphs}.
This implies that $\{w\}$ is also the only minimal hitting set of $\Hyp_{\min}$.

Via the min-heuristic, MMCS selects $\{w\}$ for branching first,
recognizes in linear time that the resulting solution $S = \{w\}$ 
is irredundant and a hitting set, and outputs the solution.
Since $w$ is the only candidate branching vertex in the branching edge $\{w\}$,
no other child node is explored and the algorithm terminates.

\begin{restatable}{lemma}{minheuristicnotpoly}
\label[lemma]{lem:min-heuristic_not_poly}
	\textup{MMCS} without min-heuristic working on input hypergraph $(V_{\min}, \Hyp_{\min})$
	requires $2^{\Omega(|V_{\min}|)} \cdot |\Hyp_{\min}|$ time
	even if violator pruning is used.
\end{restatable}

\begin{proof}
	For some index set $J \subseteq \{1, \dots, \ell\}$,
	we use the notation $U_J = \{ u_j \mid j \in J \} \sqcup \{ u'_j \mid j \in j \}$.
	We also write $\overline{J}$ for the complement $\{1, \dots, \ell\} {\setminus} J$.
	We claim that for any non-negative integer $k \le \ell$,
	the $k$-th level of the search tree contains $2^k$ nodes $(S, C)$	
	such that, for each of them, there exists an index set $J$ with $|J| = k$ 
	so that $S \subseteq U_J$ is irredundant and $C \supseteq U_{\overline J}$.
	Observe that the claim is sufficient for the lemma
	since MMCS must explore $\Omega (2^\ell) = 2^{\Omega(|V_{\min}|)}$ nodes
	and each one takes linear time in the number of edges for the redundancy checks
	and for selecting a branching edge.	
	
	We prove the claim by induction over $k$.
	In the root, $k = 0$, the only node is $(\emptyset,V)$ 
	satisfying the claim for index set $J = \emptyset$.
	Let now $(S, C)$ be a node in the $(k{-}1)$-th level for which the claim holds. 
	We investigate its child nodes.
	By induction, there exists a set $J$ of $k-1$ indices,
  	such that $S \subseteq U_J$ is irredundant and $C \supseteq U_{\overline{J}}$.
	We further have $|S| = k-1$ in level $k-1$.
	Exactly one of the two vertices $u_j$ or $u'_j$ is in $S$ for each $j \in J$,
  	resulting in
  	$\unhit(S) =  \left\{\{w, u_j, u'_j\} \mid j \in \overline{J} \right\} \sqcup \{\{w\}\}$.
  	Edges of the first kind exist as $k-1 < \ell$ implies that $\overline{J} \neq \emptyset$.
  	
  	By assumption, MMCS selects any branching edge $E \in \unhit(S)$
	for which $|E \cap C|$ is not minimum.
	The candidates $C$ contain $U_{\overline{J}}$,
	thus $|\{w, u_js, u'_j\} \cap C| \ge 2$.
	As $|\{w\} \cap C| \le 1$ has a strictly smaller number of branching candidates,
	one of the former edges is selected.
	Let $j^* \in \overline{J}$ be the respective index.
	MMCS creates (at least) the two child nodes with 
	partial solutions $S \sqcup \{u_{j^*}\}$ and $S \sqcup \{u'_{j^*}\}$.
 	Both are irredundant: $u_{j^*}$ and $u'_{j^*}$ hit only this one edge of the hypergraph,
  	implying that all vertices of $S$ keep their private edges.
	
	Which vertices end up in the respective candidate sets 
    depends on which child node is expanded first and 
    whether violator pruning is enabled.
  	In any case, MMCS removes at most the candidate vertices of the selected edge,
  	implying that the new candidate sets are both supersets of 
  	$C{\setminus}\{w, u_{j^*}, u'_{j^*}\} \supseteq U_{\overline{J \sqcup \{j^*\}}}$.
	Therefore, the two generated child nodes satisfy the conditions 
	for the index set $J \sqcup \{j^*\}$.
 	Applying this construction to each of the relevant $2^{k-1}$ nodes on the $(k{-}1)$-th level
 	proves the claim for level $k$.
\end{proof}

\subsection{Removing Violators.}
\label{subsec:old_violator}

Let $(S,C)$ be the current node and $E$ the branching edge.
A vertex $v_i \in E \cap C$ being a violator can have one of two reasons:
either $S_i = S \sqcup \{v_i\}$ is redundant or it is a hitting set.
In both cases, it makes no sense to include $v_i$ in the candidate set 
of any child node with higher index $j > i$.
In the subtree rooted at the sibling $(S_j,C_j)$,
adding $v_i$ to the partial solution $S_j$ would make it redundant all the same
(either $S_j \sqcup \{v_i\}$ contains a redundant set or it \emph{properly} contains the hitting set $S_i$).

However, for our first observation, it makes sense to distinguish the two cases
and focus only on those vertices that make partial solutions redundant.
If a branching candidate $v \in E \cap C$ 
gets re-inserted in the candidate sets of a later siblings,
we say that the descendants of that sibling \emph{inherit} the vertex $v$.
 
\begin{restatable}{lemma}{inheritance}
\label[lemma]{lem:inheritance}
	Let $(S,C)$ be a node in the search tree, $E \in \unhit(S)$ an edge,
	and $v \in E \cap C$ a branching candidate.
	If $v$ causes the partial solution $S$ to become redundant,
	then $v$ was inherited from some sibling of an ancestor of $(S,C)$.
\end{restatable}

\begin{proof}
	Whenever MMCS branches on some edge, the elements of that edge get removed from the candidates.
	Inheritance is the only way to increase the candidate sets again.
	Let $E_1, \dots, E_{|S|}$ be the branching edges corresponding 
	to the parent-child transition on the path from the root $(\emptyset,V)$ 
	to the current node $(S,C)$.
	The edge $E_j$ is a private edge for the vertex 
	that was added to $S$ in the $j$-th transition.
	Moreover, $C^* = V{\setminus} \bigcup_{j=1}^{|S|} E_j$ 
	are precisely the candidates in $C$ that are \emph{not} inherited.
	We prove that $S \sqcup \{v\}$ is irredundant for any 
	$v \in C^*$.
	By definition, $v$ is not contained in any $E_j$ and 
	thus leaves its status as a private edge intact.
	The new branching edge $E$ contains $v$ but is disjoint from $S$,
	it serves as a private edge for $v$.
\end{proof}

MMCS is traversing the search tree in a depth-first manner.
In the first branch it explores, each node is the first of its siblings to be expanded.
No inheritance takes place.
This leaves only two possible ways for this branch to end:
either the partial solution $S$ becomes a minimal hitting set
or MMCS notices that not even $S \cup C$ is a hitting set.
By \Cref{lem:min-heuristic_candidates}, the latter cannot happen
with the min-heuristic.

\begin{restatable}{corollary}{mininheritance}
\label[corollary]{cor:min-heurisitc_inheritance}
	The first branch in the \textup{MMCS} search tree when using the min-heuristic 
	always yields a minimal hitting set.
\end{restatable}

We now show that violator pruning not only improves practical running time, see~\cite{Murakami14Dualization},
but has a similarly large effect on the asymptotic behavior of MMCS as the min-heuristic.
The family of hyergraphs $\Hyp_{\text{vio}}$ that we construct 
is again parameterized by a positive integer $\ell$.
This time there are $3\ell$ vertices $V_{\text{vio}} = W \sqcup U \sqcup U'$
partitioned into three equal-sized sets $W = \{w_1, \dots, w_\ell\}$,
$U = \{u_1, \dots, u_\ell\}$, and $U' = \{u'_1, \dots, u'_\ell\}$.
There are $\ell+1$ edges.
The set $W$ is one of them;
the others contain all but three vertices of $W$ together with a matching pair from $U$ and $U'$.
In summary, we have
\begin{multline*}
	\hspace*{-.75em}\Hyp_{\text{vio}} = \left\{ \{u_j,u'_j\} \sqcup W{\setminus}\{w_{j-1}, w_{j-2}, w_{j-3}\} \mid 1 \le j \le \ell \right\}\\ \sqcup \{W\},
\end{multline*}
where the indices are modulo $\ell$.
See \Cref{fig:violator_pruning}.

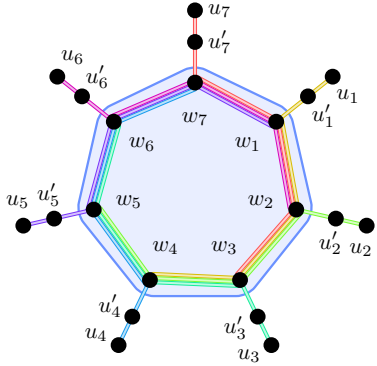
\begin{figure}
	\begin{center}
	\scalebox{0.85}{%
	\begin{tikzpicture}[baseline=-0.3em]
        \def\numvertices{7}
        \pgfmathtruncatemacro{\numverticesMinusOne}{\numvertices - 1}
        
        % 1. Main Radii 
        \pgfmathsetmacro{\rad}{max(3.0, 1.8 + (\numvertices * 0.4))}
        \pgfmathsetmacro{\midRad}{\rad + 1.8}    % u'_i ring (inner of the two outer rings)
        \pgfmathsetmacro{\outerRad}{\midRad + 1.4} % u_i ring (outermost ring)
        
        % 2. Shape Dimensions for the lanes (Reverted to outlined style)
        \def\pillOuterWidth{1.8pt}
        \def\pillInnerWidth{0.9pt}
        
        % 3. Exact Mathematical Packing for Straight Lines
        % The shift between lanes matches the outer width to guarantee zero gaps
        \pgfmathsetmacro{\laneShiftPt}{1.8 / cos(180 / \numvertices)}

        % 4. Pre-calculate baseline coordinates (1-based)
        \foreach \i in {1,...,\numvertices} {
            \pgfmathsetmacro{\ang}{90 - \i * 360 / \numvertices}
            \coordinate (w\i) at (\ang:\rad em);
            \coordinate (up\i) at (\ang:\midRad em);   
            \coordinate (u\i) at (\ang:\outerRad em);  
        }
        
        % Determine inner sweep steps (Visits l-3 nodes, meaning l-4 steps)
        \ifnum\numvertices>3
            \pgfmathtruncatemacro{\maxStep}{\numvertices - 4}
        \else
            \pgfmathtruncatemacro{\maxStep}{0}
        \fi

        % 5. Draw the filled W polygon FIRST (in the background)
        % Calculates boundary to sit 4pt completely outside the outermost lane
        \pgfmathsetmacro{\wPolyBoundPt}{(\maxStep/2.0 * \laneShiftPt) + (\pillOuterWidth/2.0) + 4.0}
        \xdef\wpath{}
        \foreach \i in {0,...,\numverticesMinusOne} {
            \pgfmathsetmacro{\ang}{90 - \i * 360 / \numvertices}
            \ifnum\i=0
                \xdef\wpath{ ($ (0,0) + (\ang : \rad em + \wPolyBoundPt pt) $) }
            \else
                \xdef\wpath{\wpath -- ($ (0,0) + (\ang : \rad em + \wPolyBoundPt pt) $) }
            \fi
        }
        % Draw the solid shape covering all W vertices
        \filldraw[fill=conceptblue!15!white, draw=conceptblue, line width=1.0pt, rounded corners=0.5em] \wpath -- cycle; % line join=bevel

        % 6. Draw the l rainbow hyperedges
        \foreach \i in {1,...,\numvertices} {
            
            % RGB Color Calculation
            \pgfmathsetmacro{\hueAngle}{\i * 360 / \numvertices}
            \pgfmathsetmacro{\valR}{((cos(\hueAngle - 0) + 1) / 2)^0.8}
            \pgfmathsetmacro{\valG}{((cos(\hueAngle - 120) + 1) / 2)^0.8}
            \pgfmathsetmacro{\valB}{((cos(\hueAngle - 240) + 1) / 2)^0.8}
            \definecolor{edgecolor}{rgb}{\valR, \valG, \valB}
            
            % Start from the outer nodes flowing inwards: u_i -> u'_i -> w...
            \xdef\edgepath{(u\i) -- (up\i)}
            
            \ifnum\numvertices>3
                \foreach \m in {0,...,\maxStep} {
                    % Target node index (Starts at i, goes to i+1, i+2... leaving out i-1, i-2, i-3)
                    \pgfmathtruncatemacro{\targetIdx}{mod(\i + \m, \numvertices)}
                    \pgfmathsetmacro{\targetAng}{90 - \targetIdx * 360 / \numvertices}
                    
                    % Calculate symmetric lane offset
                    \pgfmathsetmacro{\offsetPt}{(\maxStep/2.0 - \m) * \laneShiftPt}
                    
                    % Add straight line to the shifted coordinate
                    \xdef\edgepath{\edgepath -- ($ (0,0) + (\targetAng : \rad em + \offsetPt pt) $)}
                }
            \fi
            
            % Draw opaque border
            \draw[line width=\pillOuterWidth, line join=bevel, line cap=round, draw=edgecolor] \edgepath;
            
            % Draw simulated transparent fill
            \draw[line width=\pillInnerWidth, line join=bevel, line cap=round, draw=edgecolor!30!white] \edgepath;
        }

        % 7. Draw all nodes and labels ON TOP
        \foreach \i in {1,...,\numvertices} {
            \pgfmathsetmacro{\ang}{90 - \i * 360 / \numvertices}
            
            % w_i nodes
            \node[seqnode, fill=black] (node_w\i) at (w\i) {};
            \node[anchor=center] at ($ (w\i) + (\ang-180:1.6em) $) {$w_{\i}$};
            
            % u'_i nodes (middle ring)
            % Offset tangentially (\ang-90) so it's perfectly radially aligned with u_i
            \node[seqnode, fill=black] (node_up\i) at (up\i) {};
            \node[anchor=center] at ($ (up\i) + (\ang-90:1.1em) $) {$u'_{\i}$};
            
            % u_i nodes (outer ring)
            \node[seqnode, fill=black] (node_u\i) at (u\i) {};
            \node[anchor=center] at ($ (u\i) + (\ang-90:1.1em) $) {$u_{\i}$};
        }
    \end{tikzpicture}
    }	% end of \scalebox
    \end{center}
\caption{The hypergraph $\Hyp_{\text{vio}}$ for $\ell = 7$.
	The blue area in the center is the edge $W$.
	The other edges are drawn as paths.}
\label{fig:violator_pruning}
\end{figure}

\begin{restatable}{lemma}{violatorpolysolutions}
\label[lemma]{lem:violator_polynomial_solutions}
	Each minimal hitting set of $\Hyp_{\textup{vio}}$ has size at most $4$,
	hence there are only $O(|V_{\textup{vio}}|^4)$ solutions.
\end{restatable}

\begin{proof}
	Since $W \in \Hyp_{\text{vio}}$ is an edge,
	each minimal hitting set $T \in \Tr(\Hyp_{\text{vio}})$ contains some vertex $w_j$.
	Let $E_j = \{u_j,u'_j\} \sqcup W{\setminus}\{w_{j-1}, w_{j-2}, w_{j-3}\}$.
	The vertex $w_j$ alone already hits all edges except for
	$E_{j+1}$, $E_{j+2}$ and $E_{j+3}$.
	Since all vertices in $T$ must have their own private edge, we have $|T| \le 4$.
\end{proof}

\begin{restatable}{lemma}{violatorpolytime}
\label[lemma]{lem:violator_polynomial_time}
	\textup{MMCS} in default configuration on input $\Hyp_{\textup{vio}}$
	constructs a search tree of depth at most 11.	
	In particular, the algorithm enumerates the minimal hitting sets of $\Hyp_{\textup{vio}}$
	in polynomial time.
\end{restatable}

\begin{proof}
  As soon as a partial solution $S$ contains a vertex of $W$, 
  the subtree under the node $(S, C)$, for any $C$, has depth at most $4$
  (any vertex from $W$ hits all but three edges and every vertex added in a step of MMCS 
  hits at least one new edge).
  We show that selecting a vertex from $W$ becomes inevitable within eight levels.
  
  Consider the fourth level of the MMCS tree (with the root $(\emptyset,V)$ being the $0$-th level).
  To bound the depth of the tree, we care about the nodes with 
  an irredundant partial solution $S$ that is not yet a hitting set.
  Only they may have child nodes in the next level.
  If every such set $S$ contains a vertex from $w$, we are done.
  Let thus $(S,C)$ be a node in the fourth level such that $S$ is irredundant and disjoint from $W$.
  Moreover, without loss of generality, we assume that the next branching edge is 
  $E_j = \{u_j, u'_j\}\sqcup W{\setminus}\{w_{j-1}, w_{j-2}, w_{j-3} \}$ for some index $j$.
  If the edge were $W$, all child nodes would be guaranteed to contain a vertex from $W$.
  
  MMCS branches on $E_j \cap C$.
  To track the candidate sets of the child nodes, we observe that $S  \sqcup \{w_i\}$
  is redundant for all $i$.
  This is because the $4$ vertices in $S$ previously each had private edges, 
  while $w_i$ hits all but $3$ edges in the hypergraph.
  All branching candidates in $( W{\setminus}\{w_{j-1}, w_{j-2}, w_{j-3}\} ) \cap C$
  are violating and thus not re-inserted in the child nodes' candidate sets.
  From $W$, only the vertices $w_{j-1}$, $w_{j-2}$ and $w_{j-3}$ may remain.
  
  Iterating that argument, for any node in the fifth level whose partial solution $S$
  is irredundant and contains no vertex from $W$,
  the next branching edge is w.l.o.g.\
  $E_k = \{u_k, u'_k\}\sqcup W{\setminus}\{w_{k-1}, w_{k-2}, w_{k-3} \}$ for some $k \neq j$.
  As before, all vertices that MMCS considers for branching that are part of $W$ 
  are violating and thus removed from the candidate sets.
  Since $E_j \neq E_k$, at least one of the vertices $w_{j-1}$, $w_{j-2}$ and $w_{j-3}$ 
  that may previously have remained is now removed as a violator.
  Repeating this one more time in the sixth level,
  we are guaranteed to lose yet another of those vertices.
  
  If there is an irredundant partial solution $S$ in the seventh level containing no vertex from $W$, 	
  its candidate set $C$ has at most one vertex from $W$.
  Since the min-heuristic is used and $W \in \unhit(S)$, 
  \Cref{lem:min-heuristic_candidates} implies $|W \cap C| = 1$.
  We show that the min-heuristic will select the edge $W$ next.
  Let $E_i = \{u_i, u'_i\}\sqcup W{\setminus}\{w_{i-1}, w_{i-2}, w_{i-3} \}$
  be an unhit edge (if no such edge is unhit, there is nothing to show).
  It is the only edge in $\Hyp_{\text{vio}}$ that contains the vertices $u_i$ and $u'_i$.
  Hence, they were not removed in any parent-child transition leading to the current node $(S,C)$.
  We get
  \begin{align*}
    |E_i \cap C|
      &= | \{u_i, u'_i\} | + | (W{\setminus}\{w_{i-1}, w_{i-2}, w_{i-3}\}) \cap C \nwspace | \\
      &\ge 2 > 1 = |W \cap C| .
  \end{align*}

  In conclusion, the egde $W$ is selected 
  and all irredundant partial solutions in the eighth level must contain a vertex from $W$.
  We can conclude that the MMCS search tree has depth at~most~$11$.
\end{proof}

We show next that without violator pruning,
MMCS takes exponential time for $\Hyp_{\text{vio}}$.
The original design~\cite{Murakami14Dualization} does not specify
the order of the branching candidates in $E \cap C$.
The first one, which we denote $v_1$,
has the smallest candidate set $C{\setminus}E$ and the later ones $v_2, v_3, \dots, v_k$
potentially have larger sets through inheritance.
For the proof of \Cref{lem:violator_not_poly,lem:default_configuration_not_poly},
we assume the worst-case processing order.
This assumption is discussed thoroughly in \Cref{sec:new_heuristic}.

\begin{restatable}{lemma}{violatornotpoly}
\label[lemma]{lem:violator_not_poly}
	\textup{MMCS} using the min-heuristic but no violator pruning on input hypergraph $(V_{\textup{vio}}, \Hyp_{\textup{vio}})$
	requires $2^{\Omega(|V_{\textup{vio}}|)} \cdot |\Hyp_{\textup{vio}}|$ time
	if the branching candidates are processed in worst-case order.
\end{restatable}

\begin{proof}
	Recall the index set notation from the proof of \Cref{lem:min-heuristic_not_poly}.
	For an $J \subseteq \{1, \dots, \ell\}$, we use $U_J$
	for the set $\{u_j, u'_j \mid j \in J\}$.
	Similarly as in that proof, the lemma is implied by the following claim.
	For any non-negative integer $k \le \ell$,
	the $k$-th level of the search tree contains $2^k$ nodes $(S, C)$	
	such that, for each of them, there exists  set of $|J| = k$ indices 
	so that $S \subseteq U_J$ is irredundant and $C \supseteq U_{\overline J} \sqcup W$.
	(Note that the claim about the candidate set $C$ is more specific here and also involves $W$.)
	The claim holds vacuously for $k = 0$.
	
	We consider an arbitrary node $(S, C)$ in the $(k{-}1)$-th level for which the claim holds
	with index set $J$.
	From $|S| = |J| = k-1$, we get that $S$ contains either $u_j$ or $u'_j$ for each $j \in J$,
	hence the unhit edges are exactly $W$ as well as $E_j = \{u_j, u'_j\} \sqcup W{\setminus}\{w_{j-1}, w_{j-2}, w_{j-3}\}$ for all $j \in \overline{J}$.
  	There is an edge of the latter form because $|J| = k-1 < \ell$.
	Since the candidate set $C$ contains both $U_{\overline{J}}$ and $W$ by assumption, we know that for every $j$, the number of branching candidates is $|E_j \cap C| \le 2 + |W| - 3 = \ell-1$.
  For the edge $W$, we have $|W \cap C| = |W| = \ell$.
  Therefore, the min-heuristic selects the edge $E_{j^*}$ for some $j^* \in \overline{J}$.
  
  Note that we need the assumption $C \supseteq W$ in order for $E_{j^*}$ to be selected.
  If we had $w_{j-1}, w_{j-2}, w_{j-3} \notin C$ for some $j \in \overline{J}$
  then the edge $W$ would be selected instead because
 \begin{align*}
  	|E_j \cap C| &= |\{u_j, u'_j\}| + | \nwspace (W{\setminus}\{w_{j-1}, w_{j-2}, w_{j-3}\}) \cap C \nwspace|\\
  		&= 2 + |W \cap C| > |W \cap C|.
  \end{align*} 
  
  We now specify the order in which the branching candidates are processed.
  We assume that vertices of $U$ and $U'$ are last,
  meaning that they get the largest candidate sets.
  Since $E_{j^*}$ is selected, MMCS creates the partial solutions
  $S \cup \{u_{j^*}\}$ and $S \cup \{u'_{j^*}\}$ in the last two child nodes
  The relative order of $u_{j^*}, u'_{j^*}$ does not matter for the proof,
  we say that $u_{j^*}$ goes first to ease notation. 
  Both partial solutions are irredundant as $u_{j^*}$ and $u'_{j^*}$ hit only $E_{j^*}$.
  Since violator pruning is disabled, the candidate set
  for the second to last child (corresponding to $u_{j^*}$) is $C{\setminus}\{u_{j^*},u'_{j^*}\}$
  and it is $C{\setminus}\{u'_{j^*}\}$ for the last child.
  They contain not just $U_{\overline{J \sqcup \{j^*\}}}$ 
  but also $W$.
  The child nodes satisfy the claim.
  Without the worst-case order, the last two child nodes may get candidate sets 
  that are missing some vertices of $W$.
\end{proof}

\subsection{Default Configuration.}
\label{subsec:old_default}

One could think that min-heuristic and violator pruning together
are powerful enough to make MMCS output-polynomial.
We refute this hypothesis with a family of hard instances for MMCS in default configuration
(see \Cref{fig:default_configuration}).
For now, we use the same worst-case vertex order.
We explain in \Cref{subsec:new_adapted} how to remove this assumption.
Fix a positive integer $p \ge 2$ and let $\ell$ be a integer parameter.
This time we require that $\ell > p$.
The vertex set is $V_{\text{def}} = W \sqcup U$
where $W = \{w_1, \dots, w_\ell\}$ and $U = \{u_1, \dots, u_\ell\}$.
We use $\binom{W}{p} = \{ W' \subsetneq W \mid |W'| = p \}$ for the complete
$p$-uniform hypergraph on $W$.
The edge set in our construction is
\begin{equation*}
	\Hyp_{\text{def}} = \binom{W}{p} \sqcup \{\{u_j,w_j\} \mid 1 \le j \le \ell\}.
\end{equation*}

\begin{figure}
	\begin{center}
	\scalebox{0.85}{%
	\begin{tikzpicture}[baseline=-0.3em]
        % 0. Apply optional arguments
        \tikzset{drawthree/selected w={}, drawthree/selected u={}, drawthree/.cd}
        
        \def\numvertices{7}
        \def\uniformity{3}
        
        % CALCULATED RADIUS: Solved for R(7) = 4.01 and R(3) = 3.01
        \pgfmathsetmacro{\rad}{2.26 + (\numvertices * 0.25)}
        
        % TIGHTER OUTER RADIUS for placing the outer vertices (u_i)
        \pgfmathsetmacro{\outerRad}{\rad + 2.4}
        
        % TIGHTER PADDING to push the background boundary 
        \pgfmathsetmacro{\polySize}{2 * \rad + 1.6}
        
        % Pill shape sizes for the hyperedges {w_i, u_i}
        \def\pillOuterWidth{14.5pt}
        \def\pillInnerWidth{13pt}
        
        % Pre-calculate coordinates for inner (w_i) and outer (v_i) vertices
        \foreach \i in {1,...,\numvertices} {
            \pgfmathsetmacro{\ang}{90 - (\i-1)*360/\numvertices}
            \coordinate (w\i) at (\ang:\rad em);
            \coordinate (v\i) at (\ang:\outerRad em);
        }
        
        \ifnum\uniformity=2
            % ---------- p = 2 : EXACT GRAPH DRAWING ----------
            
            % Draw all internal \binom{W}{2} edges as normal black lines
            \foreach \i in {1,...,\numvertices} {
                \pgfmathtruncatemacro{\startJ}{\i + 1}
                \ifnum\startJ>\numvertices\else
                \foreach \j in {\startJ,...,\numvertices} {
                    \draw[seqedge, draw=black, line width=1.5pt] (w\i) -- (w\j);
                }
                \fi
            }
            
            % Draw the outer {w_i, u_i} edges as normal black lines
            \foreach \i in {1,...,\numvertices} {
                \draw[seqedge, draw=black, line width=1.5pt] (w\i) -- (v\i);
            }
            
        \else
            % ---------- p > 2 : SYMBOLIC CLOUD & PILL DRAWING ----------
            
            % Polygon rotation logic for even/odd nodes
            \pgfmathtruncatemacro{\isEven}{mod(\numvertices, 2)}
            \ifnum\isEven=0
                \pgfmathsetmacro{\polyRotation}{180 / \numvertices}
            \else
                \pgfmathsetmacro{\polyRotation}{0}
            \fi

            % 1. Draw the Background FIRST
            \ifnum\uniformity=\numvertices
                % ell = p : The center is exactly one hyperedge W.
                % Outline is black 0.75pt (matches pill border thickness), fill is gray!15
                \node[regular polygon, regular polygon sides=\numvertices, 
                      minimum size=\polySize em, 
                      rotate=-\polyRotation,
                      rounded corners=0.5em, 
                      fill=gray!13, 
                      draw=black!85, line width=0.75pt,
                      inner sep=0pt, outer sep=0pt] at (0,0) {};
            \else
                % ell > p : Symbolic cloud representing \binom{W}{p}
                \node[regular polygon, regular polygon sides=\numvertices, 
                      minimum size=\polySize em, 
                      rotate=-\polyRotation,
                      rounded corners=0.5em, 
                      fill=conceptlightblue, 
                      inner sep=0pt, outer sep=0pt] at (0,0) {};
            \fi
                  
            % 2. Draw the 2-uniform edges (pill shapes) covering {w_i, u_i}
            \foreach \i in {1,...,\numvertices} {
                \begin{scope}[transparency group, opacity=0.85]
                    % Outer black border of the pill
                    \draw[line width=\pillOuterWidth, line cap=round, draw=black] 
                        (w\i) -- (v\i);
                     
                    % Inner filled area of the pill
                    \draw[line width=\pillInnerWidth, line cap=round, draw=gray!15] 
                        (w\i) -- (v\i);
                \end{scope}
            }
        \fi

        % --- UNIFIED NODE AND LABEL DRAWING (Handles both p=2 and p>2) ---
        \foreach \i in {1,...,\numvertices} {
            \pgfmathsetmacro{\ang}{90 - (\i-1)*360/\numvertices}
            
            % --- 1. Draw and Label w_i ---
            \gdef\isWSelected{0}
            \ifx\dtSelectedW\empty\else
                \foreach \sel in \dtSelectedW {
                    \ifnum\i=\sel\relax \gdef\isWSelected{1} \fi
                }
            \fi
            
            \ifnum\isWSelected=1
                % Spiky Star for Selected Vertex (12 points)
                \node[star, star points=7, star point ratio=2.0, fill=conceptdarkblue, rounded corners=1.6pt, inner sep=0pt, minimum size=0.9em] (node_w\i) at (w\i) {};
            \else
                % Standard node
                \node[seqnode, fill=black] (node_w\i) at (w\i) {};
            \fi
            
            % Placement of w_i label depends on p
            \ifnum\uniformity=2
                \node[anchor=center] at ($ (w\i) + (\ang+55:1.1em) $) {$w_{\i}$};
            \else
                \node[anchor=center] at (\ang:\rad em - 1.4em) {$w_{\i}$};
            \fi
            
            % --- 2. Draw and Label u_i ---
            \gdef\isUSelected{0}
            \ifx\dtSelectedU\empty\else
                \foreach \sel in \dtSelectedU {
                    \ifnum\i=\sel\relax \gdef\isUSelected{1} \fi
                }
            \fi
            
            \ifnum\isUSelected=1
                % Spiky Star for Selected Vertex (12 points)
                \node[star, star points=7, star point ratio=2.0, fill=conceptdarkblue, rounded corners=1.6pt, inner sep=0pt, minimum size=0.9em] (node_v\i) at (v\i) {};
            \else
                % Standard node
                \node[seqnode, fill=black] (node_v\i) at (v\i) {};
            \fi
            
            \node[anchor=center] at (\ang:\outerRad em + 1.4em) {$u_{\i}$};
        }

    \end{tikzpicture}
    } % end of \scalebox
    \end{center}
    \caption{The hypergraph $\Hyp_{\text{def}}$ for $p \ge 3$ and $\ell = 7$.
    	The blue shaded area in the middle is the complete \mbox{$p$-uniform} hypergraph on vertex set 
    	$\{w_1, \dots, w_7\}$.}
\label{fig:default_configuration}
\end{figure}
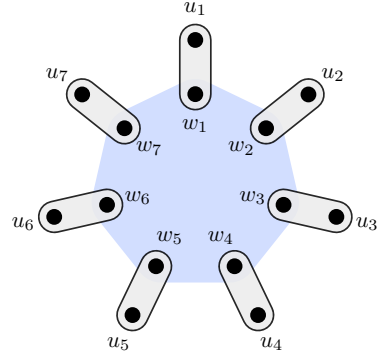

We first prove that there are only polynomially many minimal hitting sets
for bounded $p$.

\begin{restatable}{lemma}{defaultpolysolution}
\label[lemma]{lem:default_configuration_polynomial_solutions}
	The hypergraph $\Hyp_{\textup{def}}$ has $\sum_{i=0}^{p-1} \binom{\ell}{i}$ 
	minimal hitting sets, this is of order $O(|V_{\textup{def}}|^{p-1})$.
\end{restatable}

\begin{proof}
	We claim that the minimal hitting sets of $\Hyp_{\textup{def}}$
	are precisely those that take the set $W$ and replace up to $p-1$
	of the vertices $w_j$ with the corresponding vertex $u_j$.
	There are 
	$\sum_{i=0}^{p-1} \binom{\ell}{i} = O(\ell^{p-1}) = O(|V_{\textup{def}}|^{p-1})$ such sets.
	It is easy to see that any set of that form hits every edge.
	Any edge $\{u_j,w_j\}  \in \Hyp_{\text{def}}$ intersects such a set in only one vertex,
	these are the private edges.

	It is left to prove that there are no other minimal hitting sets.
	Let $T \in \Tr(\Hyp_{\text{def}})$.
	For each index $1 \le j \le \ell$ $T$ contains either $u_j$ or $w_j$.
	It must have at least one of them to hit $\{u_j,w_j\}$,
	but cannot have both as otherwise $u_j$ does not have a private edge
	($u_j$ only appears in $\{u_j,w_j\}$).	
	We also have $|T \cap U| \le p-1$.	
	If there were distinct vertices
	$u_{j_1}, \dots, u_{j_{p-1}},u_{j_p} \in T$,
	then the requirement of a private edge for each of them would result in
	$w_{j_1}, \dots, w_{j_p} \notin T$.
	If so, $T$ could not hit the edge $\{w_{j_1}, \dots, w_{j_p}\} \in \binom{W}{p}$.
\end{proof}

We first examine MMCS on $\Hyp_{\text{def}}$ for $p \ge 3$ and 
defer the discussion of $p=2$ until the end of the section.

\begin{restatable}{lemma}{defaultnotpoly}
\label[lemma]{lem:default_configuration_not_poly}
	\textup{MMCS} in default configuration on input hypergraph 
	$(V_{\textup{def}},\Hyp_{\textup{def}})$ with parameter $p \ge 3$
	requires $2^{\Omega(|V_{\textup{def}}|)} \cdot |\Hyp_{\textup{def}}|$ time
	if the branching candidates are processed in worst-case order.
\end{restatable}

\begin{proof}
	We claim that for all non-negative integers $k < \ell$,
	there are $2^k$ nodes $(S,C)$ in the $k$-th level
	that each have an index set $J$ with $|J| = k$ such that
	\begin{enumerate}
		\item $(S \cap W) \sqcup (C \cap W) = W$;
		\item for all $j \in J$, either $u_j \in S$ or $w_j \in S$;
		\item $C \supseteq U_{\overline{J}}$. 
	\end{enumerate}
	This implies that the search tree has $\Omega(2^{\ell}) = 2^{\Omega(|V_{\text{def}}|)}$
	nodes and each of them requires linear time for the redundancy check
	and finding a branching edge.	
	
	The claim holds for the root $(\emptyset,V)$.
	Let $(S,C)$ be a node in level $k-1$ for which the conditions hold with index set $J$.
	Observe that Property~2 implies that $S$ is irredundant.
	To find out which edge is selected by the min-heuristic,
	we calculate $|E \cap C|$ for an unhit edge $E \in \binom{W}{p}$.
  	None of the vertices $w \in E$ are in $S$.
    By Property 1, they are all in the candidate set $C$, 
    hence $|E \cap C| = p$.
    On the other hand, $|\{u_j, w_j\} \cap C| \le 2$.
    
    By our assumption $p \ge 3$, the min-heuristic selects an unhit edge $\{u_{j^*}, w_{j^*}\}$ 
    with ${j^*} \in \overline{J}$.
%    (it exists as $J \subsetneq \{1,\dots, \ell\}$).
    MMCS considers the partial solutions $S \sqcup \{w_{j^*}\}$ and $S \sqcup \{u_{j^*}\}$.
    We assume the same vertex order as in the proof of \Cref{lem:violator_not_poly},
    namely, that $S \sqcup \{w_{j^*}\}$ gets processed first or,
    equivalently, that $S \sqcup \{u_{j^*}\}$ 
    receives the larger candidate set $C{\setminus}\{u_{j^*}\}$.
    None of the vertices are pruned as violators as both partial solutions are irredundant 
    and non-hitting.
    Irredundancy follows from both new partial solutions satisfying Property 2 
    with index set $J \sqcup \{j^*\}$.
    Not being a hitting set is due to $|J \sqcup \{j^*\}| = k < \ell$.
	
	To conclude the induction step, we still need to check Properties~1 and 3.
    Since we removed at most $u_{j^*}$ and $w_{j^*}$ from the candidate sets, 
    both of them contain all $u_j$ and $w_i$ for $j \notin J \sqcup \{j^*\}$.
    In particular, Property~3 holds.
    The remaining condition is that every vertex of $W$
    is either in the partial solution or the corresponding candidate set.
    For the second child node $(S \sqcup \{u_{j^*}\}, C{\setminus}\{u_{j^*}\})$,
    since all vertices from $W$ are distributed the same as in the parent $(S,C)$.
    For the first child $(S \sqcup \{w_{j^*}\}, C{\setminus}\{u_{j^*}, w_{j^*}\})$
    the vertex $w_{j^*}$ is moved from the candidate set to the partial solution.
    
    The last argument relies on the vertex order.
    If the child node with partial solution $S \sqcup \{u_{j^*}\}$ were first,
    its candidate set would be $C{\setminus}\{w_{j^*},u_{j^*}\}$,
    so $w_{j^*}$ is missing from both the partial solution and the candidates,
    breaking Property~1.
    If this happens for at least $p-2$ times on a path 
    from the root to the current node,
    some edges in $\binom{W}{p}$ would become eligible for branching.
\end{proof}

\noindent
\textbf{The case $p = 2$.}
If we set $p$ to $2$ then $\Hyp_{\text{def}}$ is a graph
and its hitting sets are the vertex covers.
The only time we used the assumption on $p$ in the proof of \Cref{lem:default_configuration_not_poly}
was to argue that the min-heuristic prefers to hit all edges of the form $\{u_j,w_j\}$
before selecting any $E \in \binom{W}{p}$ for branching.
If the latter are now also of size $2$, it is not specified in \cite{Murakami14Dualization}
which one is chosen.
If we additionally assume a worst-case edge ordering (among all with the minimum number of candidates)
the same proof shows that MMCS is unable to enumerate 
minimal vertex covers in output-polynomial time.
This is surprising as there are many efficient algorithms known for that problem~\cite{Lawler80GeneratingAllMaxIndSets,MakinoUno04EnumeratingAllMaxCliques,%
Manoussakis17OutputSensitiveMaxCliqueEnumeration,Tsukiyama77GeneratingAllTheMaxIndSets}.
In fact, enumerating vertex covers (equivalently, cliques or independent sets) in graphs is
the application that lead to the concept of output-polynomiality in the first place~\cite{Johnson88MaxIndSet}.

\section{Unhit Degree Sorting.}
\label{sec:new_heuristic}

The counterexamples in the previous section depend on MMCS,
whenever it is branching on $2$ candidates,
to repeatedly process the node last that \emph{in hindsight} created the larger search tree.
This suggests that, if one suspects a vertex to disproportionally grow the tree,
processing it early may improve execution time.
We propose a heuristic implementing this idea and prove that it is indeed equivalent to minimizing
a certain property related to the tree size.

Let $E$ be the branching edge.
While all branching candidates in $E \cap C$ are tried out eventually,
the difference is that the later ones have larger candidate sets 
through inheritance from the earlier child nodes.
For the current node $(S,C)$ consider the subhypergraph $\unhit(S) \subseteq \Hyp$
of edges that are not yet hit by the partial solution.
For a vertex $v$, we let its \emph{unhit degree} 
$\deg_S(v) = | \{F \in \unhit(S) \mid v \in F\}|$ be the number of unhit edges containing $v$.

\begin{itemize}
	\item \emph{Unhit degree sorting}:
		The branching candidates are processed in order of increasing unhit degree.
\end{itemize}

We first argue that the new heuristic enables MMCS to solve $\Hyp_{\text{def}}$ efficiently
(the instances $\Hyp_{\text{vio}}$ are already solved by violator pruning).

\begin{restatable}{lemma}{defaultconfignew}
\label[lemma]{lem:default_configuration_new}
	\textup{MMCS} in default configuration and using unhit degree sorting enumerates the 
	minimal hitting sets of $\Hyp_{\textup{def}}$ in polynomial time.
\end{restatable}

\begin{proof}
	The idea of this proof is to identify a polynomial-sized collection 
	$\Syp \subseteq 2^{V_{\text{def}}}$
	and to show that \emph{each} branch in the MMCS search tree reaches a node $(S,C)$
	such that the following statements hold.
	\begin{enumerate}
		\item $S \in \Syp$.
		\item $(S,C)$ has at most $\ell-1$ ancestors.
		\item The subtree rooted at $(S,C)$ has depth at most $p$.
	\end{enumerate}	
	Constant depth of the subtree implies polynomial size.
	Combining this with the fact that no two nodes in the search tree have the same partial solution
	now shows that the \emph{total} tree size is polynomial.
	We claim this for the collection $\Syp$ of all sets $S$
	that are irredundant, contain all but exactly $p$ vertices of $W$,
	and at most $p-1$ vertices from $U$.
	Since $S$ is irredundant, for each $w_j \in S$ we have $u_j \notin S$ and vice versa.
	There are thus at most $\binom{|W|}{p} = \binom{\ell}{p}$ possibilities for $S \cap W$
    and, for each of them, at most $\sum_{i=0}^{p-1} \binom{p}{i}$ choices for $S \cap U$,
    so $|\Syp| = O(\ell^{p})$.
	
	To prove the claim, 
	we define $W_{\textup{ex}} = W{\setminus}(S \cup C)$.
	Observe that $|W_{\textup{ex}}| \le p-1$ always holds.
	Otherwise, there were an edge in $\binom{W}{p}$ that has an empty intersection with $S \cup C$,
	which is impossible by \Cref{lem:min-heuristic_candidates}.	
  
  	We split a branch of the search tree into phases.
  	The first phase lasts until the $(p{-}2)$-th branching.
  	The min-heuristic ensures that all branching edges are of the form $\{u_j,w_j\}$.
  	Since $|C \cap W| \ge \ell - |S| \ge \ell-(p-3)$, 
  	any edge in $\binom{W}{p}$ contains at least $3$ candidates
  	and is therefore not selected for branching.
  	Unhit degree sorting ensures 
  	that any time some $u_j$ is included in the partial solution,
	$w_j$ is removed from the candidates, i.e., moved to $W_{\text{ex}}$.
  	We have $\deg_S(u_j) = 1$ as the vertex appears 
  	exclusively in the current branching edge $\{u_j,w_j\}$.
	Because $|S \cap W| \le |S| \le p-3$, 
	there is at least one more unhit edge in $\binom{W}{p}$ that contains $w_j$. 	
  	We thus have $\deg_S(u_j) < 2 \le \deg_S(w_j)$.
  	The two child nodes (in order) are $(S \sqcup \{u_j\}, C{\setminus}\{u_j,w_j\})$
  	and $(S \sqcup \{w_j\}, C{\setminus}\{w_j\})$.
  	The inclusion of $u_j$ in $S$ coincides 
  	with the inclusion of $w_j$ in $W_{\text{ex}}$.
  	We use this connection explicitly below.
  	For now, it is enough that removing $w_j$ from $C$
  	ensures that $u_j$ retains its private edge, thereby guaranteeing that the branch remains irredundant.
	The end of the phase is actually reached by the branch.  	
  	
  	The second phase lasts until $|S \cap W| = \ell - p$.
  	It is now also possible that edges from $\binom{W}{p}$ are selected for branching.
  	This is not a problem since they can neither increase the number of vertices from $U$ in $S$
  	nor can they make $S$ redundant.
  	Whenever MMCS does branch on some $\{u_j,w_j\}$ in that phase,
  	the same argument as above still shows that the child nodes are ordered favorably
    via unhit degree sorting.
    At the end of the phase, $S$ is in $\Syp$:
    it is irredundant, $|S \cap U| = |W_{\text{ex}}| \le p-1$,
    and $|S \cap W| = \ell - p$ by assumption.
    
    Each branching of MMCS adds exactly one vertex to $S$.
    The inclusion $S \in \Syp$ implies $|S| \le (p-1) + (\ell-p) = \ell-1$,
    which bounds the length of the two phases.
    Moreover, each such $S$ hits all but at most $p+1$ edges 
    (those $\{u_j,w_j\}$ for which none of the two vertices are in $S$ as well as $W{\setminus}S \in \binom{W}{p}$).
    Therefore, the search tree rooted at any such node $(S,C)$ has depth at most $p$.
\end{proof}

\subsection{Minimizing the Potential Branching Points.}
\label{subsec:new_branching}

We now relate unhit degree sorting to the size of the search tree.
The connection is via a concept we call potential branching points.
Let $(S,C)$ be a node such that $S$ is irredundant but not yet a hitting set.
A \emph{potential branching point} is an edge-vertex pair $(E,v)$
such that $E \in \unhit(S)$ and $v \in E \cap C$.
If $S$ is redundant or a hitting set, we say it has no potential branching points.
The intuition is that if $E$ is selected next as the branching edge,
then the potential branching points $(E,v)$ become actual child nodes.
We are especially interested in the total number of potential branching points over all child node.
This is an upper bound on the number of grandchildren.

\pagebreak

\begin{restatable}{lemma}{potentialbranchingpoints}
\label[lemma]{lem:potential_branching_points}
	For a given node $(S,C)$ and branching edge $E \in \unhit(S)$,
	sorting the vertices in $E \cap C$ by unhit degree
	is equivalent to minimizing the number of potential branching points
	over all child nodes of $(S,C)$.
\end{restatable} 

\begin{proof}
	We prove that
	ordering the vertices in such a way that minimizes the number of potential branching points
	among the children is equivalent to ordering the non-violating vertices ascendingly 
	by their unhit degree, while violating vertices can appear at any place in the ordering.
	Clearly, a child node $(S', C')$ has fewer potential branching points 
	than the parent $(S, C)$.
  	We calculate how much smaller this number is.
  
  Of course, the child node loses all potential branching points 
  for edges $E' \in \unhit(S){\setminus}\unhit(S')$,
  i.e., those that get hit by the new vertex added to $S'$.
  However, when considering all child nodes together, this reduction in potential branching points
  does not depend on the order, all branching candidates $E \cap C$ are tried out eventually.
  A similar argument allows us to assume that there are no violators in $E \cap C$,
  they get removed from the candidate sets of all children
  regardless of the order.  
  
	The vertex order does affect the number of potential branching points 
	in the following way.
	If a vertex $v \in C$ gets excluded from $C'$, 
	all potential branching points $(E', v)$ with $v \in E'$ and $E' \in \unhit(S')$ are removed.
	The extent of this effect is precisely $\deg_{S'}(v)$.
	For concreteness, let $E \cap C = \{v_1, \dots, v_k\}$ in the order in which the branching candidates are processed.
	The partial solution of the $i$-th child node is $S' = S \sqcup \{v_i\}.$	
	By the assumption that there are no violators, its candidate set 
	is $C' = C{\setminus}\{v_i, \dots, v_k\}$.
    Since the vertices $v_i, \dots, v_k$ are excluded, the reduction 
    in potential branching points in the $i$-th child node is
%    \begin{equation*}
    	$\sum_{j=i}^k \deg_{S \sqcup \{v_i\}}(v_j)$.
%    \end{equation*}

	The other key observation is that the unhit degree of a vertex w.r.t.\ $S \sqcup \{v_i\}$
	can be expressed in terms of the unhit degree w.r.t.\ $S$ in the parent.
	For any $1 \le j \le k$, define $\mathcal{E}_{i j} = \{ E \in \unhit(S) \mid v_i, v_j \in E \}$.
	Clearly, the definition is symmetric in $i$ and $j$.
	Let 
	\begin{equation*}
		\Hyp_{\text{def}}(v_j) = \{E \in \Hyp_{\text{def}} \mid v_j \in E\}
	\end{equation*}	
	denote the edges that contain $v_j$.
	The unhit edges can be decomposed into
	\begin{equation*}
		\unhit(S) = \unhit(S \sqcup \{v_i\}) \sqcup \{ E \in \unhit(S) \mid v_i \in E \}.
	\end{equation*}
	It follows that
	\begin{multline*}
		\unhit(S) \cap \Hyp_{\text{def}}(v_j) =\\
			(\unhit(S \sqcup \{v_i\}) \cap \Hyp_{\text{def}}(v_j)) \sqcup \mathcal{E}_{ij}.
	\end{multline*}
	This implies $\deg_{S \sqcup \{v_i\}}(v_j) = \deg_{S}(v_j) - |\mathcal{E}_{ij}|$
	
	Summing over all child nodes, the total reduction $R$ of potential branching points is.
	\begin{align*}
    	R &= \sum_{i=1}^k \sum_{j=i}^k \deg_{S \sqcup \{v_i\}}(v_j)\\
    	 &= \underbrace{\sum_{i=1}^k \sum_{j=i}^k \deg_{S}(v_j)}_\textrm{(I)} - 
    	 	\underbrace{\sum_{i=1}^k \sum_{j=i}^k |\mathcal{E}_{ij}|}_\textrm{(II)}
	\end{align*}	
  By symmetry, the term $\textrm{(II)}$ does not depend on the ordering of vertices.
  Maximizing the reduction $R$ is thus the same as maximizing $\textrm{(I)}$.
  We rewrite that term as
  \begin{equation*}
  	\textrm{(I)} = \deg_{S}(v_1) + 2 \deg_S(v_2) + \cdots + k \cdot \deg_{S}(v_k).	
  \end{equation*}
  This representation shows that the last vertex $v_k$ should have the largest unhit degree,
  $v_{k-1}$ the second largest and so on.
  In summary, unhit degree sorting minimizes the potential branching points.
\end{proof}

\subsection{Adapted Counterexample.}
\label{subsec:new_adapted}

Unfortunately, unhit degree sorting is not enough to make MMCS output-polynomial.
The counterexample in \Cref{lem:default_configuration_not_poly}
can be adapted to instances $\Hyp_{\text{ad}}$ that also
force exponential running time in the presence of the new heuristic,
while keeping the total number of solutions polynomial.
We present here the main differences to the construction in \Cref{subsec:old_default}.
Let $q \ge \max(3,p-1)$ be a fixed positive integer (independent of $\ell$).
Besides $U$ and $W$, we add another set \(X = \{x_1, \dots, x_\ell\}\) to get the vertex set $V_{\text{ad}}$.
We introduce a hyperedge between every vertex \(u_j \in U\)
and every subset of \(q\) vertices from \(X\).
%The total edge set is
\begin{align*}
	\Hyp_{\text{ad}} = \Hyp_{\text{def}} 
		\sqcup \left\{ \{u_j\} \sqcup X_0 \mid 1 \le j \le \ell, X_0 \in \binom{X}{q} \right\}
\end{align*}

To prove that the total number of minimal solutions is still polynomial,
we use what is commonly known as Berge's algorithm~\cite{Berge89Hypergraphs}.
It states that for any hypergraph $\Hyp$ and decomposition 
into (not necessarily disjoint) subhypergraphs $\Hyp = \Gyp \cup \Fyp$, we have\footnote{%
	In fact, it holds that 
	$\Tr(\Hyp)$ contains exactly the minimal edges from
	$\{T_{\Gyp} \cup T_{\Fyp} \mid T_{\Gyp} \in \Tr(\Gyp), T_{\Fyp} \in \Tr(\Fyp)\})$, see~\cite{Berge89Hypergraphs}.
}
\begin{equation*}
	\Tr(\Hyp) \subseteq \{T_{\Gyp} \cup T_{\Fyp} \mid T_{\Gyp} \in \Tr(\Gyp), T_{\Fyp} \in \Tr(\Fyp)\}).
\end{equation*}
It is thus enough to demonstrate that the subhypergraph
$\left\{ \{u_j\} \sqcup X_0 \mid 1 \le j \le \ell, X_0 \in \binom{X}{q} \right\}$
has polynomially many minimal hitting sets.
Those are either the whole set $U$ or any set in $\binom{X}{\ell-q+1}$,
$O(\ell^{q-1})$ sets in total. 

We now show that the new construction forces MMCS to explore an exponential search tree.
First, we note that MMCS does not branch on any of the new hyperedges until all
of the form $\{u_j, w_j\}$ are hit.
In that phase the candidate set contains \(X\)
and hence for any hyperedge \(\{u_j\} \cup X_0\),
the intersection with the candidate set contains at least \(q \ge 3\) vertices.

The key argument in the proof of \Cref{lem:default_configuration_not_poly},
was that, when branching on $\{u_j,w_j\}$,
the vertex $w_j$ is processed first.
Therefore, in the other child node corresponding to $u_j$,
$w_j$ is not excluded from the candidate set.
The new edges \(\{u_j\} \cup X_0\) increase the unhit degree of $u_j$ to 
at least \(1 + \binom{\ell}{q}\).
Conversely, any vertex in \(W\) at any time has at most \(1 + \binom{\ell-1}{p-1}\) 
incident hyperedges (unhit or otherwise).
Since $q \ge p-1$, unhit degree sorting prefers $w_j$ over $u_j$.

\section{Experiments.}
\label{sec:eval}

We now present the results of our empirical analysis of the new heuristic.
Beside unhit degree sorting,
we also examine another approach based on the regular degree of a vertex
in the whole input.

\begin{itemize}
	\item \emph{Global degree heuristic}:
		The branching candidates are processed in order of increasing degree.
\end{itemize}

\subsection{Implementation Details.}
\label{subsec:eval_implementation_details}

We briefly describe here how the MMCS implementation conducts the redundancy checks
and how it verifies whether the current partial solution $S$ is a hitting set.
The algorithm maintains a marked list of all edges of the input hypergraph $\Hyp$.
The marks depend on the interplay with $S$:
an edge $E$ is marked \texttt{unhit} if $|E \cap S| = 0$, 
\texttt{critical} if $|E \cap S| = 1$, and \texttt{other} for $|E \cap S| > 1$.
A \texttt{critical} edge stores the (unique) vertex for which it is critical.

The partial solution $S$ itself is organized as a stack supporting the depth-first search in the tree,
the last vertex that got added is the first to be removed during backtracking.
Each vertex on the stack stores all edge marks that changed during its inclusion.
This allows for efficient unwinding when the vertex is popped.
Note that any vertex $v$ affects only the marks of the $\deg(v)$ edges incident to $v$.
Also, each vertex has a counter for the number of edges that are critical for it.
They get updated whenever a \texttt{critical} mark changes to \texttt{other} or vice versa.
For the redundancy check, it is thus enough to verify 
that every $v \in S$ has a positive counter.

Finally, the unhit edges in $\unhit(S)$ are maintained as an explicit list
together with its length.
Checking whether $S$ is a hitting set thus comes down to $|\unhit(S)| = 0$.
This list is only updated lazily via the edge marks
whenever MMCS is looking for the next branching edge.
In particular, when a node in the tree is made a leaf,
$\unhit(S)$ does not need to be updated.

These data structures also support the implementation of the new heuristics.
When using unhit degree sorting, we go through the edges incident to a vertex $v_i \in |E \cap C|$
and count the \texttt{unhit} marks.
The global degree heuristic can be implement by a one-off preprocessing in the beginning.
After reading the input, we sort the vertices by degree.
The name of any vertex during the execution is its rank in the sorting.
Before the output, the naming of the vertices in each minimal solution in $\Tr(\Hyp)$
is reversed to match the initial input.

\subsection{Experimental Setup.}
\label{subsec:eval_setup}

All experiments were run on a 16-core AMD Ryzen Threadripper PRO Zen5 9955WX processor (\qtyrange[range-phrase=--]{4.5}{5.4}{\giga\hertz}, x86\_64)
with \qty{64}{\giga\byte} of DDR5 ECC memory (\qty{6400}{\mega\hertz}),
operated under Ubuntu 24.04.4 LTS running the 6.17.0-40-generic kernel.
The entire code is written in Rust and was compiled with Cargo version 1.96.1 in release mode.
The evaluation was isolated to a single CPU core.
The code and testing data is available at \href{https://anonymous.4open.science/r/mmcs-is-not-output-polynomial-1F10}{anonymous.4open.science/r/mmcs-is-not-output-polynomial-1F10}.

We warmed up the system for 5 minutes on unrelated inputs before the experiments,
but kept the caches cold by never running the same configuration or the same hypergraph twice in a row.
Running time measurements were done as the median of 5 runs per instance and configuration
(3 runs for very high running times) to smooth outside influence.
As expected for a deterministic algorithm, 
the variance over those runs was very low throughout.
The breakdown of the time spent in different parts of the algorithm was measured in separate runs
since the frequent measurements incurred a noticeable overhead.

For our experiments, we used 75 hypergraphs
with \numrange{13}{3340} vertices and \numrange{6}{1704511} hyperedges.
As the main measure of instance size, we use the sum over all edge cardinalities, denoted as $\|\Hyp\|$.
The hypergraphs were collected by the authors of \cite{Birnick20HPIValid,Gainer17AlgorithmsAndComputation,Murakami14Dualization}
from various sources \cite{geurts_profiling_2003,pfeiffer_metatool_1999,samaga_logic_2009,VeraLicona13OCSANA,zheng_real_2001}
covering several different real-world applications of hitting set enumeration.
A complete breakdown by application domain can be found in \Cref{app:instances}.
Of the instances, 19 are modeling winning and losing states of the game Connect~4 with different edge cutoffs, see \Cref{tab:hypergraphs-winning-c4,tab:hypergraphs-losing-c4}.
Secondly, there are 15 hypergraphs from infrequent itemset mining with different thresholds, see \Cref{tab:hypergraphs-bms-webview2,tab:hypergraphs-accidents}.
We use five hypergraphs to compute minimal cut sets in different metabolic networks 
of Escherichia coli bacteria (\cref{tab:hypergraphs-e-coli})
and three hypergraphs for computing interventions in cell signaling network 
(\cref{tab:hypergraphs-oscana-egfr}).
Finally, we also include 17 hypergraphs from the discovery of minimal unique column combinations (UCCs) in relational databases.
We translate the databases into the hypergraph of difference sets between any pair of rows
via the algorithm in \cite{Birnick20HPIValid}.
That algorithm simultaneously computes the minimal hitting sets of the difference sets
(that is, the actual UCCs). 
We also use those transversal hypergraphs for our experiments, see \cref{tab:hypergraphs-databases}.
Note that the minimal hitting sets of a transversal hypergraph are exactly the inclusion-wise minimal edges of the original hypergraph~\cite{Berge89Hypergraphs}.

As mentioned in \Cref{sec:intro}, MMCS is known 
to be very efficient on real-world data~\cite{Gainer17AlgorithmsAndComputation}.
We are mainly interested in the additional performance improvements gained from the degree heuristics.
We therefore sometimes consider only the hypergraphs where the default configuration takes at least 1 second.
We refer to these 25 instances as \emph{long-running}.

\subsection{Evaluation.}
\label{subsec:_eval}

\begin{figure}
    \centering
	\includegraphics[page=1, width=\columnwidth]{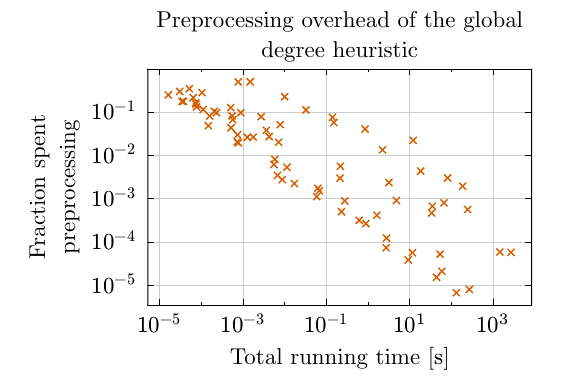}
    \caption{Log-log plot of the ratio of the preprocessing time on the total running time for the global degree heuristic. Each point represents one input hypergraph.}
    \label{fig:preprocessing}
\end{figure}

As anticipated, the unhit degree sorting reduces the size of the search tree significantly,
with the tree size being\footnote{
	Here and throughout, when reporting ratios, we refer to the geometric mean over all instances.
}
only \qty{73}{\percent} compared to the default configuration.
With tree size we mean all partial solutions considered by the algorithm,
including the redundant ones.
However, the overhead of computing the unhit degrees and sorting the branching candidates
in every node largely undermines the savings from the reduced tree size.
The improvement in running time is only \qty{9}{\percent} 
(and \qty{0.3}{\percent} in the median over all instances).

To reduce this overhead, we also consider the simpler global degree heuristic.
It has the benefit, that we can sort the vertices once at preprocessing time.
Renaming the vertices obviates the need to sort any branching edge individually.
The global degree heuristic reduces the tree size significantly,
though not as much as the unhit degree heuristic.
It is at \qty{76}{\percent} of the default configuration.
The running time, however, is notably smaller than that at \qty{70}{\percent}.
If we consider only long-running instances (running time of $\ge$\qty{1}{s} in default configuration)
the percentage is \qty{49}{\percent}.
The global degree heuristic effectively cuts the running time in half.\footnote{The most significant
improvements occur on the Connect 4 hypergraphs, where the running time is reduced to just \qty{9.7}{\percent} in one case.
Even if these instances are excluded, the running time remains at \qty{85}{\percent} for long-running instances.}

\begin{figure}
	\centering
	\includegraphics[page=5, width=\columnwidth]{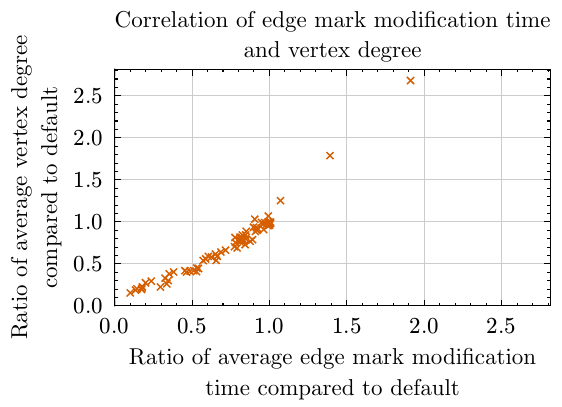}
	\caption{Correlation of edge mark modification time and vertex degree
		for the global degree heuristic.
		A data point $(x, y)$ represents a hypergraph for which
		the average time per node is reduced to fraction $x$ compared to the default configuration
		and average degree of vertices in the partial solutions is reduced to $y$.
	}
	\label{fig:edge-marks-vs-avg-degree}
\end{figure}

A more detailed breakdown shows that the preprocessing 
required by the heuristic takes just \qty{0.7}{\percent} of the running time overall,
and \qty{0.03}{\percent} on the long-running instances.
\Cref{fig:preprocessing} shows how the relative overhead of preprocessing shrinks as the total running time increases.
Note that postprocessing of reversing the vertex renaming has very little overhead, since it can be done on-the-fly while
copying a discovered minimal hitting set from the internal data structures to the output.
Although this step takes \qty{14}{\percent} longer, the total time spent providing discovered solutions to the output
remains negligible, typically less than the preprocessing.

Still, the gap between the reduction in tree size and the reduction in running time is surprising.
It can only be explained if also the time spend per node in the search tree is smaller when using
the global degree heuristic.
The work done in a node ${(S \sqcup \{v\}, C)}$ 
consists in updating (and later reverting) the edge marks and finding a new branching edge $E$.
For the former, the algorithm iterates over all edges incident to $v$;
for the latter, it iterates over the list $\unhit(S)$.
Our experiments reveal that the average degree of the vertices that get included in a partial solution
over the course of the enumeration
is indeed smaller than for the default configuration.
This reduction in the average degree explains \qty{86}{\percent} of the improvement in the
average time per node required for updating and reverting the edge marks ($R^2$ coefficient).
The correlation is displayed in \cref{fig:edge-marks-vs-avg-degree}.
Similarly, the average number of unhit edges $|\unhit(S)|$ is reduced
and explains \qty{95}{\percent} of the improvement in the average per node required for finding the branching edge,
see \cref{fig:best-edge-vs-unhit-count}.

It remains to explain \emph{why} the global degree heuristic reduces 
the average degree and number of unhit edges.
Regarding the vertex degree, the global degree heuristic
first includes vertices of small degrees,
which causes the high-degree vertices to be excluded 
from most of the candidate sets of the child nodes.
For the number of unhit edges, we note that a smaller search tree has fewer redundant nodes.
While the only other kind of leaf node, i.e.\ a minimal hitting set,
always has exactly zero unhit edges,
a redundant node can have any number of unhit edges.
Thus, reducing the number of redundant nodes in the search tree
also decreases the average number of unhit edges per node.

\begin{figure}
	\centering
	\includegraphics[page=6, width=\columnwidth]{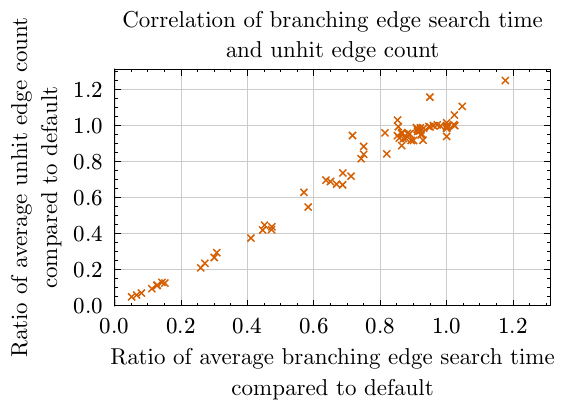}
	\caption{Correlation of the time to find a branching edge via the min-heuristic
		and unhit edge count for the global degree heuristic.
	}
	\label{fig:best-edge-vs-unhit-count}
\end{figure}

To compare the new heuristics and the default configuration against each other, we provide performance profiles \cite{Dolan02Benchmarking}.
For each heuristic, the line shows the fraction of the input hypergraphs for which the
heuristic obtains a running time (respectively tree size) that is within a certain factor of the best-achieved
running time (tree size) on the same hypergraph.
For example, the values at $x=1$ indicate the fraction of hypergraphs on which the given heuristic
achieves the best running time (tree size) out of any of the heuristics.
A line towards the top left is considered better.
In \cref{fig:perfs-combined-all-graphs},
we compare running times and tree sizes.
Additionally, \cref{fig:perf-runtime-long-running-graphs}
provides a comparison restricted to the 25 long-running input hypergraphs.

\begin{figure}[ht]
	\centering
	\includegraphics[page=2, width=\columnwidth]{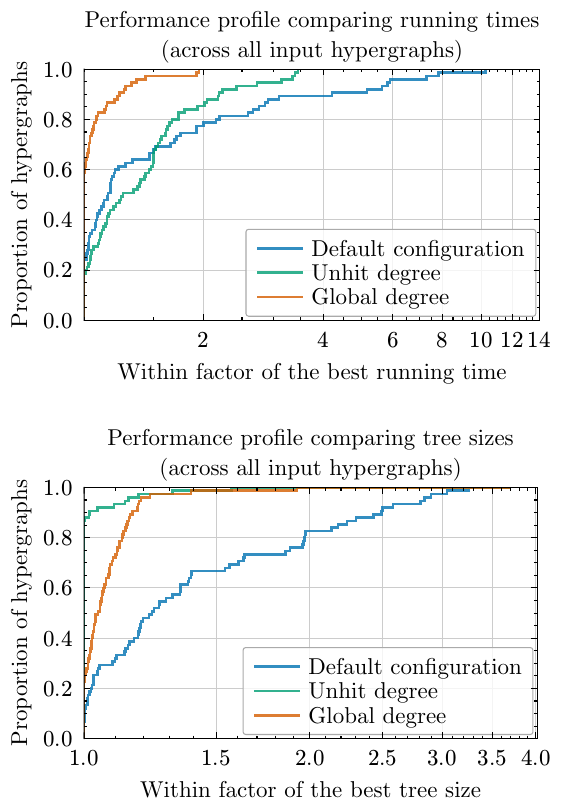}
	\caption{Performance profiles comparing running times (search tree sizes) across all input hypergraphs.}
	\label{fig:perfs-combined-all-graphs}
\end{figure}

\begin{figure}[ht]
	\centering
	\includegraphics[page=3, width=\columnwidth]{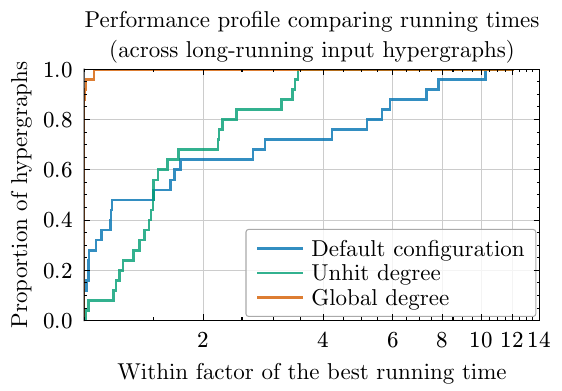}
	\caption{Performance profile comparing running times across long-running input hypergraphs.}
	\label{fig:perf-runtime-long-running-graphs}
\end{figure}

% \begin{figure}
% 	\centering
% 	\includegraphics[page=4, width=\columnwidth]{evaluation.pdf}
% 	\caption{Performance profile comparing search tree sizes across long-running input hypergraphs.}
% 	\label{fig:perf-tree-size-long-running-graphs}
% \end{figure}

\section*{Acknowledgments.}

We thank Max Göttlicher for letting us use
his implementation of MMCS in default configuration,
as well as for the proof of \Cref{lem:min-heuristic_candidates}. 
We are also thankful to Thomas Bläsius for suggesting the idea of an degree-based heuristic.

This work received support from the German Research Foundation (DFG)
under grant agreement No.~556899211 ``Design, Analysis, and Engineering of Enumeration Algorithms''.

\pagebreak

\bibliographystyle{siamplain}
\bibliography{enum_refs}

\begin{thebibliography}{10}

\bibitem{Berge89Hypergraphs}
{\sc C.~Berge}, {\em {Hypergraphs - Combinatorics of Finite Sets}}, vol.~45 of
  North-Holland Mathematical Library, North-Holland Publishing Company,
  Amsterdam, Netherlands, 1989.

\bibitem{Birnick20HPIValid}
{\sc J.~Birnick, T.~Bläsius, T.~Friedrich, F.~Naumann, T.~Papenbrock, and
  M.~Schirneck}, {\em {Hitting Set Enumeration with Partial Information for
  Unique Column Combination Discovery}}, Proceedings of the VLDB Endowment, 13
  (2020), pp.~2270--2283, \url{https://doi.org/10.14778/3407790.3407824}.

\bibitem{Bleifuss24FDsThroughHittingSets}
{\sc T.~Bleifu{\ss}, T.~Papenbrock, T.~Bl\"{a}sius, M.~Schirneck, and
  F.~Naumann}, {\em {Discovering Functional Dependencies through Hitting Set
  Enumeration}}, Proceedings of the ACM on Management of Data, 2 (2024),
  pp.~43:1--43:24, \url{https://doi.org/10.1145/3639298}.

\bibitem{Blaesius22EfficientlyJCSS}
{\sc T.~Bläsius, T.~Friedrich, J.~Lischeid, K.~Meeks, and M.~Schirneck}, {\em
  {Efficiently Enumerating Hitting Sets of Hypergraphs Arising in Data
  Profiling}}, Journal of Computer and System Sciences, 124 (2022),
  pp.~192--213, \url{https://doi.org/10.1016/j.jcss.2021.10.002}.

\bibitem{Boros98Subimplicants}
{\sc E.~Boros, V.~Gurvich, and P.~L. Hammer}, {\em {Dual Subimplicants of
  Positive Boolean Functions}}, Optimization Methods and Software, 10 (1998),
  pp.~147--156, \url{https://doi.org/10.1080/10556789808805708}.

\bibitem{DemetrovicsThi87Antikeys}
{\sc J.~Demetrovics and V.~D. Thi}, {\em {Keys, Antikeys and Prime
  Attributes}}, Annales Universitatis Scientiarum Budapestinensis de Rolando
  E{\"o}tv{\"o}s Nominatae Sectio Computatorica, 8 (1987), pp.~35--52.

\bibitem{Dolan02Benchmarking}
{\sc E.~D. Dolan and J.~J. Moré}, {\em {Benchmarking Optimization Software
  with Performance Profiles}}, Mathematical Programming, 91 (2002),
  pp.~201--213, \url{https://doi.org/10.1007/s101070100263}.

\bibitem{FredmanKhachiyan96Dualization}
{\sc M.~L. Fredman and L.~G. Khachiyan}, {\em {On the Complexity of Dualization
  of Monotone Disjunctive Normal Forms}}, Journal of Algorithms, 21 (1996),
  pp.~618--628, \url{https://doi.org/10.1006/jagm.1996.0062}.

\bibitem{Gainer17AlgorithmsAndComputation}
{\sc A.~Gainer-Dewar and P.~Vera-Licona}, {\em {The Minimal Hitting Set
  Generation Problem: Algorithms and Computation}}, SIAM Journal on Discrete
  Mathematics, 31 (2017), pp.~63--100,
  \url{https://doi.org/10.1137/15M1055024}.

\bibitem{geurts_profiling_2003}
{\sc K.~Geurts, G.~Wets, T.~Brijs, and K.~Vanhoof}, {\em {Profiling of
  High-Frequency Accident Locations by Use of Association Rules}},
  Transportation Research Record, 1840 (2003), pp.~123--130,
  \url{https://doi.org/10.3141/1840-14}.

\bibitem{Johnson88MaxIndSet}
{\sc D.~S. Johnson, C.~H. Papadimitriou, and M.~Yannakakis}, {\em {On
  Generating All Maximal Independent Sets}}, Information Processing Letters, 27
  (1988), pp.~119--123, \url{https://doi.org/10.1016/0020-0190(88)90065-8}.

\bibitem{Kante14EnumerationMinimalDominatingSets}
{\sc M.~M. Kant{\'{e}}, V.~Limouzy, A.~Mary, and L.~Nourine}, {\em {On the
  Enumeration of Minimal Dominating Sets and Related Notions}}, SIAM Journal on
  Discrete Mathematics, 28 (2014), pp.~1916--1929,
  \url{https://doi.org/10.1137/120862612}.

\bibitem{Karp72Reducibility}
{\sc R.~M. Karp}, {\em {Reducibility Among Combinatorial Problems}}, in
  Proceedings of a Symposium on the Complexity of Computer Computations, 1972,
  pp.~85--103, \url{https://doi.org/10.1007/978-1-4684-2001-2_9}.

\bibitem{Lawler80GeneratingAllMaxIndSets}
{\sc E.~L. Lawler, J.~K. Lenstra, and A.~H.~G. Rinnooy~Kan}, {\em {Generating
  All Maximal Independent Sets: NP-Hardness and Polynomial-Time Algorithms}},
  SIAM Journal on Computing, 9 (1980), pp.~558--565,
  \url{https://doi.org/10.1137/0209042}.

\bibitem{Livshits20ApproximateDenialConstraints}
{\sc E.~Livshits, A.~Heidari, I.~F. Ilyas, and B.~Kimelfeld}, {\em {Approximate
  Denial Constraints}}, Proceedings of the VLDB Endowment, 13 (2020),
  p.~1682–1695, \url{https://doi.org/10.14778/3401960.3401966}.

\bibitem{MakinoUno04EnumeratingAllMaxCliques}
{\sc K.~Makino and T.~Uno}, {\em {New Algorithms for Enumerating All Maximal
  Cliques}}, in Proceedings of the 9th Scandinavian Workshop on Algorithm
  Theory (SWAT), 2004, pp.~260--272,
  \url{https://doi.org/10.1007/978-3-540-27810-8\_23}.

\bibitem{Mannila87Dependency}
{\sc H.~Mannila and K.-J. R{\"{a}}ih{\"{a}}}, {\em {Dependency Inference}}, in
  Proceedings of the 13th International Conference on Very Large Data Bases
  (VLDB), 1987, pp.~155--158, \url{https://www.vldb.org/conf/1987/P155.PDF}.

\bibitem{Manoussakis17OutputSensitiveMaxCliqueEnumeration}
{\sc G.~Manoussakis}, {\em {An Output Sensitive Algorithm for Maximal Clique
  Enumeration in Sparse Graphs}}, in Proceedings of the 12th International
  Symposium on Parameterized and Exact Computation (IPEC), 2017,
  pp.~27:1--27:8, \url{https://doi.org/10.4230/LIPIcs.IPEC.2017.27}.

\bibitem{Marazzi20OCSANA+}
{\sc L.~Marazzi, A.~Gainer{-}Dewar, and P.~Vera{-}Licona}, {\em {OCSANA+:
  Optimal Control and Simulation of Signaling Networks from Network Analysis}},
  Bioinformatics, 36 (2020), pp.~4960--4962,
  \url{https://doi.org/10.1093/bioinformatics/btaa625}.

\bibitem{Murakami14Dualization}
{\sc K.~Murakami and T.~Uno}, {\em {Efficient Algorithms for Dualizing
  Large-Scale Hypergraphs}}, Discrete Applied Mathematics, 170 (2014),
  pp.~83--94, \url{https://doi.org/10.1016/j.dam.2014.01.012}.

\bibitem{pfeiffer_metatool_1999}
{\sc T.~Pfeiffer, I.~Sánchez-Valdenebro, J.~Nuño, F.~Montero, and
  S.~Schuster}, {\em {METATOOL: Metatool: For Studying Metabolic Networks}},
  Bioinformatics, 15 (1999), pp.~251--257,
  \url{https://doi.org/10.1093/bioinformatics/15.3.251}.

\bibitem{Reiter87DiagnosisFirstPrinciples}
{\sc R.~Reiter}, {\em {A Theory of Diagnosis from First Principles}},
  Artificial Intelligence, 32 (1987), pp.~57--95,
  \url{https://doi.org/10.1016/0004-3702(87)90062-2}.

\bibitem{samaga_logic_2009}
{\sc R.~Samaga, J.~Saez-Rodriguez, L.~G. Alexopoulos, P.~K. Sorger, and
  S.~Klamt}, {\em {The Logic of EGFR/ErbB Signaling: Theoretical Properties and
  Analysis of High-Throughput Data}}, PLOS Computational Biology, 5 (2009),
  p.~e1000438, \url{https://doi.org/10.1371/journal.pcbi.1000438}.

\bibitem{Stavropoulos16FrequentItemsetHiding}
{\sc E.~C. Stavropoulos, V.~S. Verykios, and V.~Kagklis}, {\em {A Transversal
  Hypergraph Approach for the Frequent Itemset Hiding Problem}}, Knowledge and
  Information Systems, 47 (2016), pp.~625--645,
  \url{https://doi.org/10.1007/S10115-015-0862-3}.

\bibitem{Tomita06GeneratingAllMaximalCliques}
{\sc E.~Tomita, A.~Tanaka, and H.~Takahashi}, {\em {The Worst-Case Time
  Complexity for Generating All Maximal Cliques and Computational
  Experiments}}, Theoretical Computer Science, 363 (2006), pp.~28--42,
  \url{https://doi.org/10.1016/j.tcs.2006.06.015}.

\bibitem{Tsukiyama77GeneratingAllTheMaxIndSets}
{\sc S.~Tsukiyama, M.~Ide, H.~Ariyoshi, and I.~Shirakawa}, {\em {A New
  Algorithm for Generating All the Maximal Independent Sets}}, SIAM Journal on
  Computing, 6 (1977), pp.~505--517, \url{https://doi.org/10.1137/0206036}.

\bibitem{VeraLicona13OCSANA}
{\sc P.~Vera{-}Licona, E.~Bonnet, E.~Barillot, and A.~Y. Zinovyev}, {\em
  {OCSANA: Optimal Combinations of Interventions from Network Analysis}},
  Bioinformatics, 29 (2013), pp.~1571--1573,
  \url{10.1093/bioinformatics/btt195}.

\bibitem{Xiao22DynamicFDs}
{\sc R.~Xiao, Y.~Yuan, Z.~Tan, S.~Ma, and W.~Wang}, {\em {Dynamic Functional
  Dependency Discovery with Dynamic Hitting Set Enumeration}}, in Proceedings
  of the 38th International Conference on Data Engineering (ICDE), 2022,
  pp.~286--298, \url{https://doi.org/10.1109/ICDE53745.2022.00026}.

\bibitem{YefernyAllani18VehicleAdHocNet}
{\sc T.~Yeferny and S.~Allani}, {\em {MPC: A RSUs Deployment Strategy for
  VANET}}, International Journal of Communication Systems, 31 (2018),
  \url{https://doi.org/10.1002/dac.3712}.

\bibitem{zheng_real_2001}
{\sc Z.~Zheng, R.~Kohavi, and L.~Mason}, {\em {Real World Performance of
  Association Rule Algorithms}}, in Proceedings of the 7th ACM SIGKDD
  Conference on Knowledge Discovery and Data Mining (KDD), 2001, pp.~401--406,
  \url{https://doi.org/10.1145/502512.502572}.

\end{thebibliography}

\clearpage
\appendix

\section{Overview of Test Instances}
\label[appendix]{app:instances}

We present here the hypergraphs that we used for the experiments.
Column $|V|$ shows the number of vertices, $|\Hyp|$ the number of edges,
$\|\Hyp\|$ the sum over all edge sizes, and $|\Tr(\Hyp)|$ the number of minimal hitting sets.

\begin{table}[htbp]
	\centering
	\footnotesize
	\caption{Winning Connect 4} % From gainer-devar (where did they get it from?)
	\label{tab:hypergraphs-winning-c4}
	\begin{threeparttable}
		\begin{tabular}{rrrrr}
			\toprule
			$|V|$ & $|\mathcal{H}|$ & $\|\mathcal{H}\|$ & $|\Tr(\mathcal H)|$ \\
			\midrule
			\num{75}	& \num{100}		& \num{800}		& \num{287}			\\
			\num{77}	& \num{200}		& \num{1600}	& \num{1145}		\\
			\num{77}	& \num{400}		& \num{3200}	& \num{6069}		\\
			\num{77}	& \num{800}		& \num{6400}	& \num{11675}		\\
			\num{79}	& \num{1600}	& \num{12800}	& \num{71840}		\\
			\num{81}	& \num{3200}	& \num{25600}	& \num{459502}		\\
			\num{81}	& \num{6400}	& \num{51200}	& \num{1277933}		\\
			\num{83}	& \num{12800}	& \num{102400}	& \num{4587967}		\\
			\num{83}	& \num{25600}	& \num{204800}	& \num{11614885}	\\
			\num{84}	& \num{44473}	& \num{355784}	& \num{31111249}	\\
			\bottomrule
		\end{tabular}
		\begin{tablenotes}
			\item The smaller hypergraphs are subhypergraphs of the larger ones.
		\end{tablenotes}
	\end{threeparttable}
\end{table}

\begin{table}[htbp]
	\centering
	\footnotesize
	\caption{Losing Connect 4} % From hypergraph dualization repo
	\label{tab:hypergraphs-losing-c4}
	\begin{threeparttable}
		\begin{tabular}{rrrrr}
			\toprule
			$|V|$ & $|\mathcal{H}|$ & $\|\mathcal{H}\|$ & $|\Tr(\mathcal H)|$ \\
			\midrule
			\num{76}	& \num{100}		& \num{800}		& \num{2341}		\\
			\num{76}	& \num{200}		& \num{1600}	& \num{22760}		\\
			\num{78}	& \num{400}		& \num{3200}	& \num{33087}		\\
			\num{80}	& \num{800}		& \num{6400}	& \num{79632}		\\
			\num{80}	& \num{1600}	& \num{12800}	& \num{212761}		\\
			\num{80}	& \num{3200}	& \num{25600}	& \num{2396735}		\\
			\num{80}	& \num{6400}	& \num{51200}	& \num{4707877}		\\
			\num{84}	& \num{12800}	& \num{102400}	& \num{16405082}	\\
			\num{84}	& \num{16635}	& \num{133080}	& \num{39180611}	\\
			\bottomrule
		\end{tabular}
		\begin{tablenotes}
			\item The smaller hypergraphs are subhypergraphs of the larger ones.
		\end{tablenotes}		
	\end{threeparttable}
\end{table}

\begin{table}[htbp]
	\centering
	\footnotesize
	\caption{BMS-WebView2 (infrequent itemsets)} % From hypergraph dualization repo
	\label{tab:hypergraphs-bms-webview2}
	\begin{threeparttable}
		\begin{tabular}{rrrrr}
			\toprule
			Threshold & $|V|$ & $|\mathcal{H}|$ & $\|\mathcal{H}\|$ & $|\Tr(\mathcal H)|$ \\
			\midrule
			\num{20}	& \num{3340}	& \num{30405}	& \num{101399670}	& \num{3064937}	\\
			\num{30}	& \num{3340}	& \num{17315}	& \num{57749581}	& \num{2297560}	\\
			\num{50}	& \num{3340}	& \num{6946}	& \num{23171251}	& \num{1289303}	\\
			\num{100}	& \num{3340}	& \num{2591}	& \num{8644513}		& \num{438867}	\\
			\num{200}	& \num{3340}	& \num{823}		& \num{2746671}		& \num{89448}	\\
			\num{400}	& \num{3340}	& \num{237}		& \num{791148}		& \num{15993}	\\
			\num{800}	& \num{3340}	& \num{62}		& \num{206998}		& \num{4616}	\\
			\bottomrule
		\end{tabular}
	\end{threeparttable}
\end{table}

\begin{table}[htbp]
	\centering
	\footnotesize
	\caption{Accidents (infrequent itemsets)} % From gainer-devar (but they got it from murakami & uno)
	\label{tab:hypergraphs-accidents}
	\begin{threeparttable}
		\begin{tabular}{rrrrr}
			\toprule
			Threshold & $|V|$ & $|\mathcal{H}|$ & $\|\mathcal{H}\|$ & $|\Tr(\mathcal H)|$ \\
			\midrule
			30k		& \num{442}	& \num{135439}	& \num{58291353}	& \num{185218} \\
			50k		& \num{336}	& \num{32207}	& \num{10477277}	& \num{47137} \\
			70k		& \num{336}	& \num{10968}	& \num{3576449}		& \num{17486} \\
			90k		& \num{336}	& \num{4322}	& \num{1411815}		& \num{7617} \\
			110k	& \num{81}	& \num{2000}	& \num{144468}		& \num{3547} \\
			130k	& \num{81}	& \num{990}		& \num{72121}		& \num{1916} \\
			150k	& \num{64}	& \num{447}		& \num{25183}		& \num{1039} \\
			200k	& \num{64}	& \num{81}		& \num{4656}		& \num{253} \\
			\bottomrule
		\end{tabular}
	\end{threeparttable}
\end{table}

\begin{table}[htbp]
	\centering
	\footnotesize
	\caption{E.~coli metabolic reaction networks} % From gainer-devar
	\label{tab:hypergraphs-e-coli}
	\begin{threeparttable}
		\begin{tabular}{lrrrr}
			\toprule
			Network & $|V|$ & $|\mathcal{H}|$ & $\|\mathcal{H}\|$ & $|\Tr(\mathcal H)|$ \\
			\midrule
			Acetate		& \num{103}	& \num{266}		& \num{6300}	& \num{3363}	\\
			Glucose		& \num{104}	& \num{6387}	& \num{194046}	& \num{21001}	\\
			Glycerol	& \num{105}	& \num{2128}	& \num{57883}	& \num{25619}	\\
			Succinate	& \num{103}	& \num{932}		& \num{20823}	& \num{14136}	\\
			Combined	& \num{109}	& \num{27503}	& \num{842664}	& \num{275914}	\\
			\bottomrule
		\end{tabular}
	\end{threeparttable}
\end{table}

\begin{table}[htbp]
	\centering
	\footnotesize
	\caption{EGFR cell signaling network} % From gainer-devar
	\label{tab:hypergraphs-oscana-egfr}
	\begin{threeparttable}
		\begin{tabular}{lrrrr}
			\toprule
			Resolution & $|V|$ & $|\mathcal{H}|$ & $\|\mathcal{H}\|$ & $|\Tr(\mathcal H)|$ \\
			\midrule
			Short	& \num{50}	& \num{125}		& \num{1117}	& \num{1340}	\\
			Sub		& \num{56}	& \num{234}		& \num{2315}	& \num{11765}	\\
			All		& \num{64}	& \num{11050}	& \num{182259}	& \num{13116}	\\
			\bottomrule
		\end{tabular}
	\end{threeparttable}
\end{table}

\begin{table}[htbp]
	\centering
	\footnotesize
	\caption{Databases (UCCs)}
	\label{tab:hypergraphs-databases}
	\begin{threeparttable}
		\begin{tabular}{lrrrr}
			\toprule
			{Name} & {$|V|$} & {$|\mathcal{H}|$} & {$\|\mathcal{H}\|$} & {$|\Tr(\mathcal H)|$} \\ % & \# Rows & \# Cols \\
			\midrule
			echocardiog\_h		& \num{13}	& \num{53}		& \num{354}			& \num{72}		\\ % & 132 & 13 \\
			echocardiog\_tr		& \num{13}	& \num{72}		& \num{262}			& \num{30}		\\
			hepatitis\_h		& \num{20}	& \num{110}		& \num{816}			& \num{348}		\\ % & 155 & 20 \\
			hepatitis\_tr		& \num{20}	& \num{348}		& \num{2159}		& \num{54}		\\
			horse\_h			& \num{29}	& \num{67}		& \num{735}			& \num{253}		\\ % & 300 & 29 \\
			horse\_tr			& \num{29}	& \num{253}		& \num{2116}		& \num{39}		\\
			amalgam1\_h			& \num{87}	& \num{87}		& \num{6372}		& \num{2737}	\\ % & 51 & 87 \\
			amalgam1\_tr		& \num{87}	& \num{2737}	& \num{8783}		& \num{70}		\\
			flight\_h			& \num{109}	& \num{264}		& \num{7014}		& \num{26652}	\\ % & 1001 & 109 \\
			flight\_tr			& \num{109}	& \num{26652}	& \num{158925}		& \num{161}		\\
			fd-reduced-30\_h	& \num{30}	& \num{260}		& \num{7076}		& \num{3564}	\\ % & 250001 & 30 \\
			fd-reduced-30\_tr	& \num{30}	& \num{3564}	& \num{10692}		& \num{231}		\\
			musicbrainz\_h		& \num{100}	& \num{18}		& \num{168}			& \num{2288}	\\ % & 79569 & 100 \\
			musicbrainz\_tr		& \num{100}	& \num{2288}	& \num{14520}		& \num{7}		\\
			census\_h			& \num{42}	& \num{64}		& \num{330}			& \num{80}		\\ % & 196295 & 42 \\
			census\_tr			& \num{42}	& \num{80}		& \num{2192}		& \num{29}		\\
			struct\_sheet\_h	& \num{32}	& \num{6}		& \num{55}			& \num{167}		\\ % & 664128 & 32 \\
			struct\_sheet\_tr	& \num{32}	& \num{167}		& \num{701}			& \num{6}		\\
			ncvoter\_h			& \num{19}	& \num{116}		& \num{617}			& \num{96}		\\ % & 8060060 & 19 \\
			ncvoter\_tr			& \num{19}	& \num{96}		& \num{703}			& \num{50}		\\
			lineitem\_h			& \num{16}	& \num{644}		& \num{5606}		& \num{390}		\\ % & 6001215 & 16 \\
			lineitem\_tr		& \num{16}	& \num{390}		& \num{2135}		& \num{271}		\\
			ncvoter\_allc\_h	& \num{94}	& \num{1557}	& \num{61040}		& \num{1704511}	\\ % & 7503554 & 94 \\
			ncvoter\_allc\_tr	& \num{94}	& \num{1704511}	& \num{21396288}	& \num{986}		\\
			tpch\_h				& \num{52}	& \num{29578}	& \num{1076982}		& \num{347805}	\\ % & 6001216 & 52 \\
			tpch\_tr			& \num{52}	& \num{347805}	& \num{2552241}		& \num{21955}	\\
			uniprot\_h			& \num{223}	& \num{1461}	& \num{84018}		& \num{826}		\\ % & 539166 & 223 \\
			uniprot\_tr			& \num{223}	& \num{826}		& \num{6114}		& \num{15}		\\
			isolet\_c120\_h		& \num{120}	& \num{16551}	& \num{1934576}		& \num{266584}	\\ % & 7798 & 120 \\ % truncated from 618 cols
			isolet\_c120\_tr	& \num{120}	& \num{266584}	& \num{807997}		& \num{13821}	\\
			isolet\_c200\_h		& \num{200}	& \num{61354}	& \num{12056569}	& \num{1282903}	\\ % & 7798 & 200 \\ % truncated from 618 cols
			isolet\_c200\_tr	& \num{200}	& \num{1282903}	& \num{3936397}		& \num{54156}	\\
			isolet\_c240\_h		& \num{240}	& \num{127419}	& \num{30100588}	& \num{2594776}	\\ % & 7798 & 240 \\ % truncated from 618 cols
			\bottomrule
		\end{tabular}
	\end{threeparttable}
	\begin{tablenotes}
		\item A name ending on \_h indicates a hypergraph of difference sets of a relational database,
		while \_tr are its minimal unique column combinations (hitting sets).
	\end{tablenotes}
\end{table}

\end{document}